\documentclass{article}
\usepackage{supertech}
\usepackage{graphicx}
\usepackage{amsmath}
\usepackage{amssymb}
\usepackage{hyperref}
\usepackage{breakurl}
\usepackage{psfrag}       %latex text in .eps files
\usepackage{subfigure}
\usepackage[ruled]{algorithm2e}
\usepackage{color}
\usepackage{xspace}
\newcommand{\parhead}[1]{\noindent{\textbf{#1.}\xspace}}
\newcommand{\BS}{station\xspace}
\newcommand{\BSs}{stations\xspace}
\newcommand{\SA}{Station Assignment\xspace}
\newcommand{\name}{performance\xspace}
\usepackage[normalem]{ulem}
\newcommand{\mig}[1]{\textcolor{blue}{#1}}
\newcommand{\pru}[1]{\textcolor{red}{#1}}
\renewcommand{\mig}[1]{#1}
\renewcommand{\pru}[1]{#1}
\newcommand{\toskip}[1]{}

\begin{document}

%\title{Insert your title here%\thanks{Grants or other notes
%%about the article that should go on the front page should be
%%placed here. General acknowledgments should be placed at the end of the article.}
%}
\title{Station Assignment with Reallocation\thanks{A preliminary version of this work appeared in SEA 2015~\cite{MosteiroRW15}.}}



%\author{First Author         \and
%        Second Author %etc.
%}

%\author{
%Austin Halper
%\and Miguel A. Mosteiro
%\and Yulia Rossikova
%\and Prudence W.H. Wong
%}

\author{	
		Austin Halper \thanks{
		Pace University,
		Computer Science Dept.,
		New York, NY, USA,
		\texttt{ah17939n@pace.edu}} 
	\and
		Miguel A. Mosteiro \thanks{
		Pace University,
		Computer Science Dept.,
		New York, NY, USA,
		\texttt{mmosteiro@pace.edu}} 
	\and
		Yulia Rossikova \thanks{
		Kean University,
		Computer Science Dept.,
		Union, NJ, USA, 
		\texttt{rossikoy@kean.edu}} 
	\and
		Prudence W. H. Wong \thanks{
		University of Liverpool,
		Dept. of Computer Science,
		Liverpool, UK, 
		\texttt{pwong@liverpool.ac.uk}}
}

%\authorrunning{A. Halper, M. A. Mosteiro, Y. Rossikova, and P. W. H. Wong} % if too long for running head

%\institute{F. Author \at
%              first address \\
%              Tel.: +123-45-678910\\
%              Fax: +123-45-678910\\
%              \email{fauthor@example.com}           %  \\
%%             \emph{Present address:} of F. Author  %  if needed

%\institute{	
%		Austin Halper \at
%		Pace University\\
%		Computer Science Dept.\\
%		New York, NY, USA\\ 
%		\email{ah17939n@pace.edu} 
%	\and
%		Miguel A. Mosteiro \at
%		Pace University\\
%		Computer Science Dept.\\
%		New York, NY, USA\\ 
%		\email{mmosteiro@pace.edu} 
%	\and
%		Yulia Rossikova \at
%		Kean University\\
%		Computer Science Dept.\\
%		Union, NJ, USA\\ 
%		\email{rossikoy@kean.edu} 
%	\and
%		Prudence W.H. Wong \at
%		University of Liverpool\\
%		Dept. of Computer Science\\
%		Liverpool, UK\\ 
%		\email{pwong@liverpool.ac.uk}
%}

\date{}

\maketitle

\begin{abstract}
%Insert your abstract here. Include keywords, PACS and mathematical
%subject classification numbers as needed.
We study a dynamic allocation problem that arises in various scenarios where 
mobile clients joining and leaving the system have to communicate with static stations via radio transmissions. Restrictions are a maximum delay, or laxity, between consecutive client transmissions and a maximum bandwidth that a station can share among its clients. We study the problem of assigning clients to stations so that every client transmits to some station, satisfying those restrictions. We consider reallocation algorithms, where clients are revealed at its arrival time, the departure time is unknown until they leave, and clients may be reallocated to another station, but at a cost proportional to the reciprocal of the client's laxity. We present negative results for previous related protocols that motivate the study; we introduce new protocols that expound trade-offs between station usage and reallocation cost; 
%we determine experimentally a classification of the clients that best balances those orthogonal goals; 
we determine experimentally a classification of the clients attempting to balance those opposite goals;
we prove theoretically bounds on our performance metrics; and we show through simulations that, for realistic scenarios, our protocols behave much better than our theoretical guarantees.
% \keywords{First keyword \and Second keyword \and More}
%\keywords{Base Station Assignment \and Reallocation Algorithms \and Competitive Analysis \and Radio Networks}
\end{abstract}

%\section{Introduction}
%\label{intro}
%Your text comes here. Separate text sections with
%\section{Section title}
%\label{sec:1}
%Text with citations \cite{RefB} and \cite{RefJ}.
%\subsection{Subsection title}
%\label{sec:2}
%as required. Don't forget to give each section
%and subsection a unique label (see Sect.~\ref{sec:1}).
%\paragraph{Paragraph headings} Use paragraph headings as needed.
%\begin{equation}
%\end{equation}
%
%% For one-column wide figures use
%\begin{figure}
%% Use the relevant command to insert your figure file.
%% For example, with the graphicx package use
%  \includegraphics{example.eps}
%\end{figure}
%%
%% For two-column wide figures use
%\begin{figure*}
%% Use the relevant command to insert your figure file.
%% For example, with the graphicx package use
%  \includegraphics[width=0.75\textwidth]{example.eps}
%\end{figure*}
%\begin{table}
%\begin{tabular}{lll}
%\hline\noalign{\smallskip}
%\end{tabular}
%\end{table}

%!TEX root = ./BSreallocJEArevised.tex

\section{Introduction}

We study a dynamic allocation problem that arises in various scenarios 
where data on mobile devices 
has to be gathered and uploaded periodically to one of the many
static access points available~\footnote{We consider an upstream model, but the same results apply to downstream communication.}.
Examples include \emph{wearable health-monitoring systems}, where ambulatory patients
carry physiological sensors and the data gathered must be periodically uploaded,
and \emph{participatory sensing}~\pru{\cite{KhanXAA13,RestucciaDP16}}, where communities of mobile device users upload
periodically information about their environment.
\pru{For example, in the SPA system~\cite{ShaZS+08}, sensors are attached to
participants periodically sampling the heart rate, blood pressure, movement etc.;
while in the MobGeoSen application~\cite{KanjoBPCFWCW08}, mobile phones update periodically
their geo-location and associated environment.
Depending on individuals the frequency different participants need to communicate may differ,
e.g., depending on the health conditions.}
% \marginpar{\pru{I can't find ref 15 in \cite{KhanXAA13}}}

Mobile devices, called \emph{clients}, join and leave the system continuously,
and they communicate with the static access points, called \emph{stations},
via radio transmissions.
The ephemeral nature
of the clients is modeled by characterizing each client with a \emph{life interval}
(from its arrival time to departure time), 
during which the client has to communicate with some station periodically.
The need of periodic communication is modeled by the client's \emph{laxity},
which bounds the maximum duration a client is not transmitting to some stations.
% During its life interval, the client may communicate with different stations. 
The intrinsically shared nature of the access to stations is modeled by a maximum
\emph{station bandwidth} shared among its connected clients,
by a \emph{client bandwidth} required for each transmission,
and by the client laxity governing how often it must connect to some stations.

Based on the above model, we 
%\pru{propose to} \mig{[MM: are you strong on using ``propose''? I know it's common practice, but in general I don't like to say that we propose to do something that we actually do.]} 
study the problem of assigning clients to stations so that
every client transmits to some stations satisfying the laxity and bandwidth constraints.
We consider settings where clients are revealed at its arrival time and their departure time is only revealed when they depart (as in online algorithms).
Clients may be reassigned from one station to another and
we call such reassignment \emph{reallocation}.
As to be further elaborated in the next paragraph, 
reallocation has been considered in a similar context in 
the Windows Scheduling problem~\cite{Farach-ColtonLMT14},
where the cost of reallocation is proportional to the number of clients reallocated.
While counting the number of clients reallocated ensures that we do not reallocate too much, 
this may not be a fair cost and it is typical in scheduling to 
consider reallocation (or migration) in terms of the sizes of the jobs instead of the number, e.g.,~\cite{SandersSS09}.
Intuitively reallocation causes more disturbance to a client with small laxity.
Therefore, we assume reallocation incurs a cost inversely proportional to a client's laxity~\footnote{As a first step we consider a reallocation cost in terms of laxity. 
It is of interest to consider bandwidth in the cost and we leave this future work.}.
\pru{Reallocation usually involves handover from one station to another incurring a cost that is time related and also signal related~\cite{CominardiGBO17}}.

We aim to reduce the number of active stations \pru{(a station is active if it has at least one client
allocated to it to transmit)} and reduce the reallocation cost.
However, these two goals are orthogonal,
e.g., we can reallocate the clients every time a client arrives/departs so that the number
of active stations is minimized while incurring a very high reallocation cost;
alternatively we can keep the reallocation cost to zero but we may use many active
stations after a sequence of client departures.
%In this paper, we aim to obtain a balance between the two performance metric.
In this paper, we quantify the trade-off between both performance metrics: number of active stations and reallocation cost.
We call this problem \emph{Station Assignment Problem with Reallocation} (SA). 

\parhead{Previous work}
To the best of our knowledge, the closest work to the present paper is~\cite{Farach-ColtonLMT14}, where reallocation algorithms were presented for Windows Scheduling (WS).
The WS problem~\cite{bar2003windows,chan2005temporary,bar2007windows,Farach-ColtonLMT14} is a particular case of SA where the bandwidth requirement of each client is the same and each channel (a.k.a. station in our case) can only serve one client at a time.
WS has applications to various areas such as communication networks, supply chain, job scheduling, media on demand systems, etc. 
In~\cite{Farach-ColtonLMT14}, a unit cost is incurred for each client reallocated
and the objective is to minimize an aggregate sum reflecting the amortized reallocation cost and the number of channels used.
A protocol called Classified Reallocation is showed to guarantee an amortized constant number of reallocations.
This protocol is also evaluated experimentally together with two other protocols Preemptive Reallocation and Lazy Reallocation.
%\textcolor{red}{A number of protocols have been analyzed and evaluated experimentally to achieve constant amortized reallocations with close to optimal channel usage.}
%[*Choose between blue or red text*]
%\mig{(Miguel: I think I prefer the blue text, but I'm not strong on this.)}

WS~\cite{bar2003windows,chan2005temporary,bar2007windows} was first studied without reallocation and the objective was mainly to minimize the number of channels.
As pointed out in~\cite{JacobsL14}, the WS problem can be shown to be NP-hard
by assembling results available in literature~\cite{NPh,bar2007windows,HolteMR+89}.
For the static case~\cite{bar2003windows,bar2007windows} where a client never departs, we can have online algorithm whose number of channels is only an additive of $O(\sqrt{H})$ from the optimal $H$, where $H$ is the sum of reciprocal laxities of all clients~\cite{bar2007windows}.
For the dynamic case~\cite{chan2005temporary} where a client may depart, the maximum number of channels used over time by the online algorithm is at most a constant times that of the optimal~\cite{chan2005temporary}.
This means that the comparison is against peak load which may occur at different time in the online algorithm and the optimal offline algorithm.
In~\cite{Farach-ColtonLMT14} and this work, we compare against current load.

As noted in~\cite{bar2007windows}, WS is closely related to the classical bin packing problem~\cite{CoffmanGMV1998,CoffmanGS1996,coffman2013bin}.
In addition to this, introducing bandwidth in our model gives another perspective in relation to bin packing.
If all clients have very large laxity (such that the laxity constraint does not restrict them from being assigned to the same station)
and the only concern becomes the bandwidth, then the problem of minimizing the number of stations becomes the same as minimizing 
the number of bins.
Therefore, lower bounds on the approximation ratios of bin packing, i.e., 1.54037 for asymptotic approximation ratio~\cite{BaloghBG10}
and 1.5 for absolute approximation ratio~\cite{Epstein10}, apply to the station usage ratio of our problem when reallocation is not allowed.

\parhead{SA and other assignment problems}
SA generalizes several problems.
It generalizes the WS problem %~\cite{bar2003windows,chan2005temporary,bar2007windows,Farach-ColtonLMT14} 
that considered periodic transmission to capture bandwidth sharing. 
Different objectives are considered, in~\cite{bar2003windows,chan2005temporary,bar2007windows}
the goal is to minimize the number of channels used while
in~\cite{Farach-ColtonLMT14} the goal is to minimize a combined cost of the number of reallocated clients and number of channels.
We extend the later cost function such that the number of reallocated clients is weighted inversely by the client laxity.
The problem in~\cite{AntaKMW13} considers clients with the same laxity and characterizes
adversarial arrivals that admit feasible solutions.
This makes the problem substantially different from ours as the periodic transmission can be handled as if the bandwidth is shared equally among the clients.
We generalize the study to allow different laxities, and provide trade-off between
reallocation cost and number of stations.

Our problem differs from existing scheduling problems despite sharing similarities.
SA shares the idea of assigning tasks of different bandwidth 
to stations as the load balancing problem~\cite{Aza97} of assigning jobs of different loads to machines,
yet the load balancing problem does not consider periodic transmission, does not allow reallocation,
and the objective is to minimize the maximum load.
Interval coloring~\cite{AdamyE03,EpsteinEL09} concerns the number of machines used but not periodic tasks.
Periodic tasks have been considered in real time scheduling~\cite{BG04} but the periodic appearance of the tasks is determined by the input, while in our problem the periodic appearance is determined by the algorithm to satisfy the laxity constraint.
The SA problem is also related to online assignment problems such as 
$b$-matching~\cite{pruhsBmatching}, %no reallocation and maximize requests served. 
fractional matching~\cite{azarFracmatching}, %only one request at a time and maximize throughput.
and adwords~\cite{muthuAdwords}. %only one impression at a time and maximize the number of assigned impressions.
Among other details, the objective function is different. 

We consider two orthogonal objectives which is common in scheduling context.
E.g., in energy efficient scheduling problems, one would minimize the use of energy to provide acceptable quality of service.
There are two typical approaches of optimization:
to minimize the summation of two costs, e.g., energy efficient flow time scheduling minimizes the sum of energy usage and total flow time of the tasks~\cite{AlF07}; and
to formulate two performance ratios as we do in this work, e.g., energy efficient throughput scheduling derives online algorithm that is $t$-throughput-competitive and $e$-energy-competitive~\cite{ChanCLLMW09}.
%We adopt the latter approach.
Moreover, jointly targeting high bandwidth and low delay is also quite common in practice.
%\pru{[PW: say something more?]} \mig{[MM:pending]}
For instance, in~\cite{JietalToN2015}, the authors present a greedy scheduling policy for wireless networks aimed to achieve provably good performance in terms of both, throughput and delay. The model is different from ours (multiple radio channels, which can be viewed as a discrete version of our continuous-bandwidth allocation, but only one base station and only one packet per client), but the two-dimensional optimization is the same.

Our objective function takes into account the assignment cost, which is often the optimization criteria in scheduling and network design problems.
A good example is energy efficient speed-scaling scheduling where the speed of a processor is scalable to a higher speed consuming more energy while more productive.
In~\cite{BansalCP13} the objective function is the energy usage (modeled as an arbitrary power function) plus fractional weighted flow time.
This is generalized in~\cite{PruhsPrimalDual2012} to parallel machines where the objective function is energy plus an arbitrary assignment cost.
Similar cost functions have been considered for the minimum-cost network-design problem, where
packets have to be routed through a network of speed scalable routers, and the goal is to minimize the aggregate cost of assigning a packet to a link and the energy consumption of supporting the current load on the router~\cite{DBLP:conf/focs/AndrewsAZ10}.
On the other hand, scheduling in wireless networks with reallocation of resources has also been considered~\cite{leonardi} yet
reallocation is assumed to incur no cost.

Reallocation has been considered in the context of scheduling~\cite{bender2013reallocation,SandersSS04,AlbersH12}.
In~\cite{bender2013reallocation}, a distinction is made between reassignment within server (reschedule) and between servers (migration).
Here, we assume rescheduling within a station is free and we use ``reallocation'' to refer to reassignment to other stations.
It is often that the number/size of jobs reallocated is bounded, but by different quantities,
e.g., by a function of the number of jobs in the system~\cite{bender2013reallocation}, the size of the arriving job~\cite{SandersSS04} or the number of machines~\cite{AlbersH12}.
In our problem, we bound the reallocation by the weight (cumulative inverse laxity) of the clients departed.

%!TEX root = ./BSreallocJEArevised.tex

\section{Our Results}

%In this work, we study reallocation algorithms for the \SA problem. 
%For a system where all clients have the same laxity and cannot be reallocated, previous work~\cite{AntaKMW13} has considered SA with a focus on characterizing adversarial arrivals that admit feasible solutions.  For WS, which is a particular case of SA, previous work assumed that clients do not leave~\cite{bar2003windows}, cannot be reallocated~\cite{chan2005temporary}, or the reallocation cost is constant~\cite{Farach-ColtonLMT14}. 
In this paper, we study reallocation algorithms for SA assuming that clients have %arbitrary 
laxity and bandwidth requirements (arbitrary for the analysis, set to specific values for experimental evaluation), that clients depart from the system at arbitrary times, and that they may be reallocated, but at some cost proportional to the resources needed. Specifically, our contributions are the following.
\begin{itemize}
\item We define a characterization of SA reallocation algorithms, which we call $(\alpha,\beta)$-\name, as a combination of the competitive ratio on station usage ($\alpha$) and the cost of reallocations contrasted with the resources released by departures ($\beta$).
\item We show a sequence of negative results proving that worst-case guarantees cannot be provided by previous protocols Classified Reallocation and Preemptive Reallocation~\cite{Farach-ColtonLMT14}, even if they are modified to our reallocation cost function. 
\item We present a novel SA protocol called Classified Preemptive Reallocation (CPR) where clients are \emph{classified} according to laxity and bandwidth requirements, and upon departures the remaining clients are \emph{preemptively} reallocated to minimize station usage, but only within their class. The protocol presented includes a range of classifications that exposes trade-offs between reallocation cost and station usage. 
%In fact, we found first experimentally what is the classification function that better balances these goals, and then we provided theoretical guarantees for all functions.
In fact, we first found experimentally what is the classification function that seems to balance these goals (i.e. neither of the number of active stations nor the reallocation cost is the largest observed), and then we provided theoretical guarantees for all functions considered.
\item In our main theorem, we prove bounds on both of our performance metrics, and we instantiate those bounds into three classifications and for specific scenarios in two corollaries (refer to Section~\ref{section:analysis} for the specific bounds.)
\item Finally, we present the results of our extensive simulations that allowed us to find the function that maintains both, station usage and reallocation cost, below the maximum observed. Additionally, our simulations show that, for a variety of realistic scenarios,  CPR performs better than expected by the worst-case theoretical analysis, and close to optimal on average.
\end{itemize}

%!TEX root = ./BSreallocJEArevised.tex

\section{Definitions}

\parhead{Model}
We consider a set $S$ of \BSs and a set $C$ of clients.
Each client must transmit packets to some \BS. 
Time is slotted so that each time slot is long enough to transmit one packet.
A client can be assigned to transmit to only one \BS in any given time slot.
Starting from some initial time slot $1$, we refer to the infinite sequence of time slots $1,2,3,\dots$ as \defn{global time}.
Each client $c \in C$ is characterized by
an \defn{arrival time} $a_c$ and a \defn{departure time} $d_c$, that define a \defn{life interval}  $\tau_c=[a_c,d_c]$ in which $c$ is \defn{active}. That is, client $c$ is active from the beginning of time slot $a_c$ up to the end of time slot $d_c$.
We define $C(t)\subseteq C$ to be the set of clients that are active during time slot $t$.
With respect to resources required, each client $c$ is characterized by 
a \defn{bandwidth} requirement $b_c$, 
%and a \defn{laxity}  $0 < w_c \leq |\tau_c|$, such that $c$ must transmit to some \BS in $S$ at least one packet within each 
and a \defn{laxity} $w_c$, such that $0 < w_c \leq |\tau_c|$. I.e., $c$ must transmit to some \BS in $S$ at least one packet within each 
$w_c$ consecutive time slots in $\tau_c$~\footnote{To maintain station usage low, we will assume that the laxity can be relaxed during reallocation.}.
On the other hand, each \BS $s\in S$ is characterized by a \defn{station bandwidth} or \defn{capacity} $B$, which is the maximum aggregated bandwidth of clients that \emph{may} transmit to $s$ in each time slot.

%\begin{itemize}
%\item 
%An \defn{arrival time} $a_c$ and a \defn{departure time} $d_c$, that define a \defn{life interval}  $\tau_c=[a_c,d_c]$ in which $c$ is \defn{active}. That is, client $c$ is active from the beginning of time slot $a_c$ up to the end of time slot $d_c$.
%\item
%A \defn{\BS group} $S_c \subseteq S$ of \BSs to which $c$ may transmit packets.
%\item
%A \defn{laxity}  $0 < w_c \leq |\tau_c|$, such that $c$ must transmit to some \BS in $S_c$ at least one packet within each $w_c$ consecutive time slots in $\tau_c$. 
%\textcolor{red}{assume size $2^i$?}
%\item
%A \defn{bandwidth} requirement $b_c$.
%\end{itemize}

%Each \BS $s\in S$ is characterized as follows.
%\begin{itemize}
%\item 
%a \defn{capacity} $B_s$, which is the maximum aggregated bandwidth of clients that may transmit to $s$ in each time slot. 
%\end{itemize}

%Finally, each pair of \BS $s\in S$ and client $c$ is characterized as follows.
%\begin{itemize}
%\item
%an \defn{assignment cost} $\rho$, which is the cost of assigning $c$, either for the first time or by reallocation, to  $s$.
%\item
%a \defn{transmission cost} $\tau$, which is the cost of transmitting one packet from $c$ to $s$.
%\end{itemize}

\parhead{Notation} 
Let the \defn{schedule} of a client $c$ be an infinite sequence $\sigma_c$ of values from the alphabet $\{0\}\cup S$.
Let $\sigma_c(t)$ be the $t^{th}$ value of $\sigma_c$.
A \defn{\BS assignment} is a set $\sigma$ of schedules that models the transmissions from clients to \BSs. 
That is,
for each client $c\in C$ and time slot $t$, 
it is $\sigma_c(t)=s$ if $c$ is scheduled to transmit to \BS $s\in S$ in time slot $t$, and $\sigma_c(t)=0$ if $c$ does not transmit in time slot $t$. If a client $c$ is scheduled to transmit to a station $s$ we say that $c$ is \defn{assigned} to station $s$.
Note that a client is assigned to a station 
from its arrival time or when it is reallocated to this station 
until its departure time or when it is reallocated to another station
(not only at the instant time that it transmits).
We say that a station that has clients assigned is \defn{active}, and \defn{inactive} or \defn{empty} otherwise.

%Given a \BS assignment $\sigma$, 
%let $S_{\sigma_c}=\{s| \exists t : \sigma_c(t)=s\}$ (the set of \BSs in the schedule $\sigma_c$ of client $c$),
%let $C_{s,t}=\{c| \sigma_c(t)=s\}$ (the set of clients that transmit to \BS $s$ at time $t$).
%For each time slot $t$, let $C(t)=\{c| c\in C \land a_c\leq t \leq d_c \}$ (the set of active clients at time $t$). 

\parhead{Problem} 
The \defn{\SA problem (SA)} is defined as follows.
%\marginpar{we call it assignment for literature consistency} 
For a given set of \BSs and set of clients, obtain a \BS assignment such that
(i) each client transmits to some \BS at least once within each period of length its laxity during its life interval,
(ii) in each time slot, no \BS receives from clients whose aggregated bandwidth is more than the station capacity.
%, where $\alpha\geq1$ is called the \defn{capacity stretch}.
%and (iii) the cost of the assignments plus the cost of transmissions is minimized. 
%\marginpar{is this cost function appropriate?}
%Formally, 
%\begin{itemize}
%\item Input:
%
%A set $S$ of tuples $\langle s,B_s\rangle$ as defined.
%
%A set $C$ of tuples $\langle c,a_c, d_c, w_c, S_c, b_c \rangle$ as defined.
%
%%A capacity stretch $\alpha\geq1$.
%%, an assignment cost $\rho$, and a transmission cost $\tau$.
%
%\item Output:
%
%A set of schedules $\sigma=\{\sigma_c|c\in C\}$ such that 
%%minimizes $\rho\sum_{c\in C} |S_{\sigma_c}| + \tau\sum_{s\in S}\sum_t |C_{s,t}|$ under the constraint:
%
%$\forall c\in C : \forall t\in [a_c,d_c-w_c+1] : 
%%\sum_{i=t}^{t+w_c-1} \sigma_c(t)>0$ 
%\exists i\in[t,t+w_c) : \sigma_c(i)\neq 0$
%
%$\forall s\in S : \forall t : \sum_{c\in C:\sigma_c(t)=s} b_c \leq B_s$.
%
%%where $C_{s,t}=\{c| \}$ is the set of clients that transmit to \BS $s$ at time $t$.
%
%
%\end{itemize}
Notice that, for any finite set of \BSs, there are sets of clients such that the SA problem is not solvable. We assume in this work that $S$ is infinite and what we want to minimize is the number of \emph{active} \BSs.

\parhead{Algorithms} 
We study \defn{reallocation algorithms} for SA. 
%That is, each client $c$ is revealed to the algorithm only at arrival time $a_c$ (as in online algorithms). 
That is, the parameters $w_c$ and  $b_c$ needed to assign the client to some station are revealed at time $a_c$, but the departure time $d_c$ is unknown to the algorithm until the client actually leaves the system (as in online algorithms). Then, at the beginning of time slot $t$, an SA reallocation algorithm returns the transmission schedules of all clients that are active in time slot $t$, possibly reassigning some clients from one station to another. (I.e., the schedules of clients that were already active may be changed from one time slot to another.)
We refer to the reassignment of one client as a \defn{reallocation},
%~\footnote{Previous work on reallocation algorithms for job scheduling~\cite{bender2013reallocation} make a distinction between reassignment within server (reallocation) and between servers (migration). Here we assume the reassignment within a station to be free, hence we only use ``reallocation''.} 
whereas all the reassignments that happen at the beginning of the same time slot are called a \defn{reallocation event}.
%We assume that the time taken to compute the schedules is negligible.

% in $C(t)$, where $C(t)=\{c| c\in C \land a_c\leq t \leq d_c \}$ is the set of active clients at time $t$. 

\parhead{Performance Metric}
Previous work~\cite{Farach-ColtonLMT14} has considered the number of clients reallocated as the reallocation cost.
In the present work, we consider a different scenario where the cost of reallocating a client is proportional to resources requested by that client.
%Specifically, we assume a cost for the reallocation of each client $c$ of $\rho/w_c$, where $\rho>0$ is a parameter.
Specifically, we assume a cost for the reallocation of each client $c$ of $\rho/w_c$, where $\rho>0$ is a scaling factor that generalizes this cost to different settings. For our simulations, we set $\rho=1$, since $\rho$ is also a multiplicative factor in our reallocation metric and, hence, does not provide additional information about the performance of our protocols in terms of reallocation.

Then, letting $\mathcal{R}(ALG,t)$ be the cost of the reallocation event incurred by algorithm $ALG$ at time~$t$, and $R(ALG,t)$ be the set of clients being reallocated, the overall cost is the following. 
\begin{align}
\mathcal{R}(ALG,t) &= \rho\sum_{c\in R(ALG,t)}\frac{1}{w_c}.\label{eq:cost}
\end{align}
We will drop the specification of the algorithm whenever clear from the context.

With respect to performance, we aim for algorithms with low reallocation cost and small number of active stations. Unfortunately, these are contradictory goals. Indeed, the reallocation cost could be zero if no client is reallocated (online algorithm), but the number of active stations could be as big as the number of active clients (e.g. initially multiple clients assigned to each station, and then all but one client from each active station depart). 
%On the other hand, the number of active stations could be minimized applying an offline algorithm on each time slot, but the reallocation cost could be large. 
On the other hand, the number of active stations could possibly be reduced by applying an offline algorithm on each time slot, but the reallocation cost could be large. 
Thus, we characterize algorithms with both metrics as follows.

For any SA algorithm $ALG$, let $S(ALG,t)$ be the number of active stations at time $t$ in the schedule,
let $D(ALG,t)$ be the set of clients departed since the last reallocation up to time $t$.
Denoting $\sum_{c\in C'} 1/w_c$ as the \defn{weight} of the clients in $C'\subseteq C$, let $\mathcal{D}(ALG,t)$ be the weight of the clients departed since the last reallocation up to time $t$, that is, 
$$\mathcal{D}(ALG,t)=\sum_{c\in D(ALG,t)}\frac{1}{w_c}.$$ 
Also, we denote the minimum number of active stations needed at time $t$ as $S(OPT,t)$.
Throughout, we will drop the specification of the algorithm whenever it is clear from the context.
Then, we say that an SA reallocation algorithm $ALG$ achieves an \defn{$(\alpha,\beta)$-\name} if the following holds for any input.
\begin{align*}
\max_{t} \frac{S(ALG,t)}{S(OPT,t)}\leq \alpha\\
\max_{t:\mathcal{R}(ALG,t)>0} \frac{\mathcal{R}(ALG,t)}{\mathcal{D}(ALG,t)} \leq \beta.
\end{align*}

% \pru{[PW: mention WS has close to optimal results with no reallocation and no departure. We want to see effect of departure and focus on space released by departure. so only compare to D().]}

In words, the overhead on the number of stations used by $ALG$ is never more than a multiplicative factor $\alpha$ over the optimal, and the reallocation cost, amortized on the ``space'' left available by departing clients is never more than $\beta$. 
The reallocation cost is only measured at the time when $ALG$ reallocates
some clients, i.e., when $\mathcal{R}(ALG,t)>0$, because 
it is not meaningful to consider times in between reallocation events.
The rationale of comparing $\mathcal{R}(ALG,t)$ against $\mathcal{D}(ALG,t)$ is as follows.
When clients do not depart, the WS problem admits very good approximation performance even without reallocation
(recall in the introduction that in such case there is online algorithm that differs from the optimal offline algorithm
by only an additive term~\cite{bar2007windows}).
Therefore, we are motivated to study how algorithms may benefit from reallocation when there is departure 
by reusing the space released %freed 
by the departure.

\toskip{
We amortize the reallocation cost on the weight of departed clients, 
because such weight quantifies how much the allocation could be improved.
We could also amortize on the weight of arrived clients, 
but such value does not give \pru{an} indication of whether 
an allocation can be improved or not. 
That is, for the same weight of arrived clients, 
some online allocations may be improved significantly but 
others \pru{may be} already optimal. 
Thus, in our metric reallocation cost is taken in absolute terms for arrivals and amortized for departures.
}

% \mig{
% % % , whereas the weight of arrived clients does not give any indication of whether the allocation can be improved or not. (Moreover, if we consider the aggregated weight of clients involved in \emph{all} events (arrivals and departures) the performance ratio would be even better.)
% }

Notice that the above ratios are strong guarantees, in the sense that they are the maximum of the ratios instead of the ratio of the maxima. (This distinction was called previously in the literature \emph{against current load} versus \emph{against peak load} respectively.) Moreover, the reallocation ratio computed as the maximum \emph{over reallocation events} is also stronger than the ratio of cumulative weights since the system started.

%!TEX root = ./BSreallocAlgorithmica.tex

\section{Algorithms}

%\textcolor{red}{The following applies only to WS. Perhaps we want to present it first as WS for clarity and the extensions to SA without proof for brevity.}

\parhead{Broadcast Trees}
A common theme in WS algorithms with \emph{periodic} transmission schedules is to represent those schedules with \emph{Broadcast Trees}~\cite{chan2005temporary,bar2003windows,Farach-ColtonLMT14}. Broadcast trees are a convenient representation because they allow to visualize easily how the laxities are combined. Consider for instance two clients $a$ and $b$, both with laxity $2$. Both clients may be assigned to the same station alternating their transmissions. This assignment is represented by one binary tree where $a$ and $b$ hang from the root of a broadcast tree, modeling such station schedule.
Throughout the paper, we refer to a set of broadcast trees as the \defn{forest}, and to the distance in edges from a node to the root of a broadcast tree as the \defn{depth}.
Generalizing, the $2^d$ nodes at depth $d$ in a complete binary tree represent the time slots $t \mod 2^d$ (see Figure~\ref{fig:treeSchedule}).
Then, to indicate that some (periodic) time slot has been reserved for a client $c$ to transmit to a given station $s$, we say informally that $c$ is assigned to the corresponding node in the broadcast tree of $s$. Throughout the rest of the paper, we use both indistinctively.

Notice that once a client $c$ is assigned to a node $i$, no other client can be assigned as a subtree of $i$, because all the time slots represented by $i$ have been reserved for $c$. (Refer to Figure~\ref{fig:treeOneClient}) However, sibling clients are possible because they represent interleaving reservations (as in the example with $a$ and $b$ in the previous paragraph). Thus, if at any internal node only one child has a client assigned, an empty leaf is placed in the other child, making explicit the availability of the corresponding (periodic) time slot.
Consequently, in broadcast trees all nodes have exactly zero or all possible children. 
Consider for instance the tree shown in Figure~\ref{fig:treeSomeClients}, where black nodes represent clients assigned and white nodes represent available slots. The transmission schedule in this example is depicted in the figure.
%In general, broadcast trees are not necessarily binary, if the laxities are not powers of $2$. That is, the assignment of a client $c$ with laxity $w_c=r^\ell$ is represented placing the client in an $r$-ary tree node at depth $\ell$ or above. 
Refer to~\cite{chan2005temporary,bar2003windows} for further details on broadcast trees.

\begin{figure}[t]
\begin{center}
\psfrag{tmod}{$t\mod 4\equiv1$}
\psfrag{all}{$1,2,3,\dots$}
\psfrag{allodd}{$1,3,5,\dots$}
\psfrag{alleven}{$2,4,6,\dots$}
\psfrag{s1}{$1,5,9,\dots$}
\psfrag{s2}{$2,6,10,\dots$}
\psfrag{s3}{$3,7,11,\dots$}
\psfrag{s4}{$4,8,12,\dots$}
\psfrag{a}{$c$}\psfrag{b}{$b$}\psfrag{c}{$a$}
\psfrag{1}{$1$}\psfrag{2}{$2$}\psfrag{3}{$3$}
\psfrag{4}{$4$}\psfrag{5}{$5$}\psfrag{6}{$6$}
\psfrag{7}{$7$}\psfrag{8}{$8$}\psfrag{9}{$9$}
\psfrag{10}{$10$}\psfrag{11}{$11$}\psfrag{12}{$12$}
\psfrag{13}{$13$}\psfrag{dots}{$\dots$}
\subfigure[Mapping node - time-slot.]{
\includegraphics[width=0.31\textwidth]{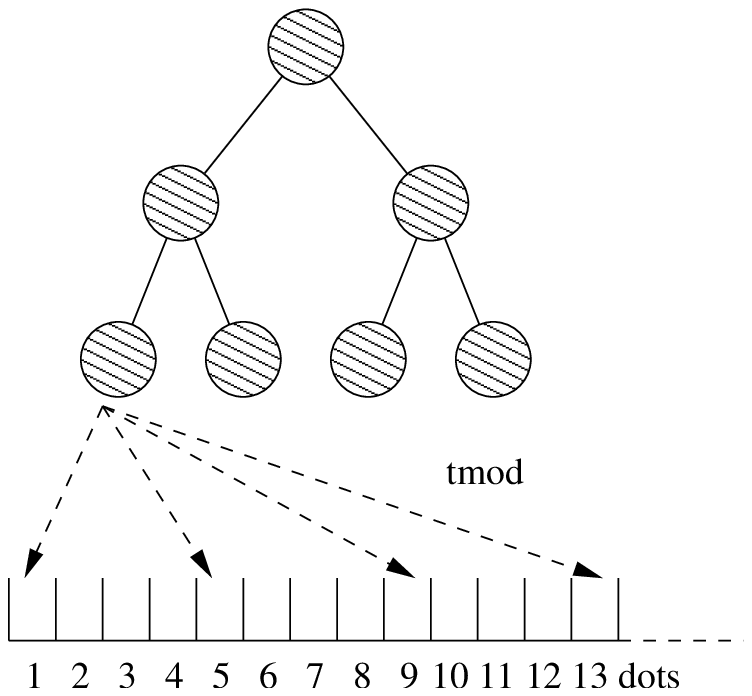}
\label{fig:treeSchedule}
}
\hfill
\subfigure[First client assigned.]{
\includegraphics[width=0.31\textwidth]{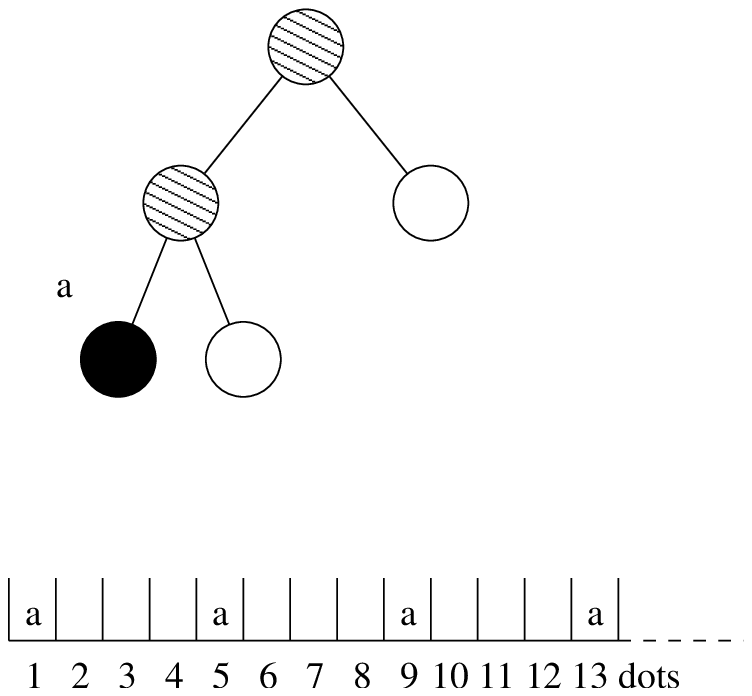}
\label{fig:treeOneClient}
}
\hfill
\subfigure[Some clients assigned.]{
\includegraphics[width=0.31\textwidth]{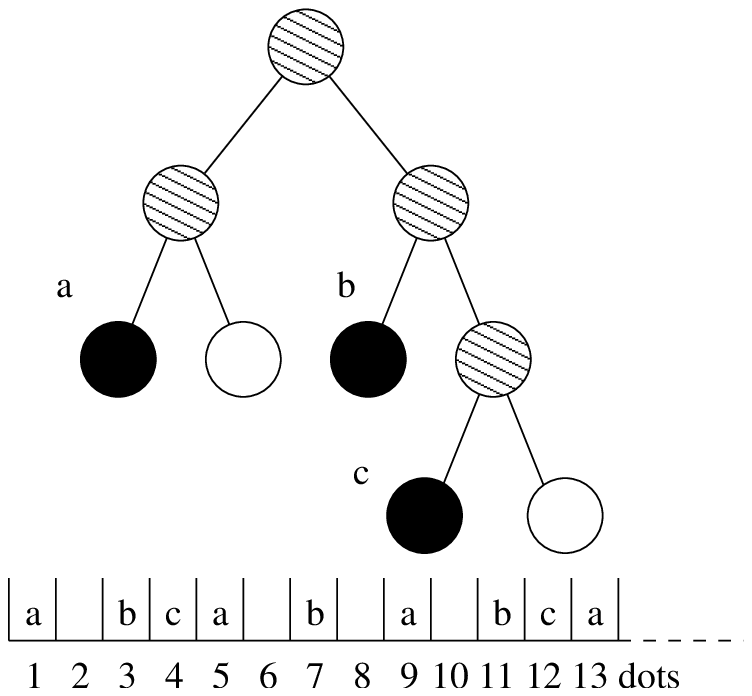}
\label{fig:treeSomeClients}
}
\caption{Illustration of a binary broadcast tree.
(a) A depth-2 tree corresponds to periodic broadcast of period $2^2$.
(b) Clients are assigned to leaves, e.g., client $c$ with laxity $4$ is assigned the black node meaning time slot $1, 5, 9$, etc.\ are reserved for it.
(c) Open leaf (white node) corresponds to available slot.} 
%\vspace{-20pt}
\end{center}
\end{figure}

\parhead{WS algorithms}
In~\cite{chan2005temporary} Chan et al. presented a WS algorithm that allocates clients with laxities that are powers of $2$ preserving the following invariant. For each station, the broadcast tree modeling the station schedule has at most one available leaf at each depth. In order to preserve this invariant, when a client departs from a tree, the remaining clients in the same tree are rearranged for free. 
%(see~\cite{chan2005temporary} for details). 
This invariant allows to upper bound the space available at each tree, but if reallocations among trees are possible, the same idea can be extended to all trees simultaneously. Indeed, that is the approach followed in the algorithm \emph{Preemptive Reallocation} (PR)~\cite{Farach-ColtonLMT14}, maintaining the invariant that  \emph{throughout all trees} there is at most one available leaf at each depth. For laxities that are powers of $2$, PR achieves an optimal station usage of $H(C(t))$ for time slot $t$, where $H(C(t))=\lceil\sum_{c\in C(t)}1/w_c\rceil$, because the sum of all empty leaves (i.e., the sum of the inverse of laxities of all clients that could be placed in those leaves) is less than $1$. Such guarantee is met re-establishing the invariant each time a client departs, possibly through reallocations among trees, at a constant cost per client reallocated between trees (within the tree are still free). It was shown experimentally that for various inputs the number of clients reallocated, amortized on the number of arrivals and departures, is constant~\cite{Farach-ColtonLMT14}. 
However, we show in Lemma~\ref{lemma:PRunbounded} that there are arrival/departure schedules for which the amortized cost in PR is unbounded.
Furthermore, 
we show in Lemma~\ref{lemma:PRexponential} that if we simply modify PR to reallocate the sibling subtree of smaller weight (rather than the subtree with less clients) to restore the invariant, there are arrival schedules for which the reallocation-cost ratio is exponential for our cost function (Equation~\ref{eq:cost}). 

A WS algorithm with provable bounded reallocation cost guarantees was shown also in~\cite{Farach-ColtonLMT14}. The protocol, called \emph{Classified Reallocation} (CR), guarantees that all clients assigned to the same station have the same laxity, except for one distinguished station that handles all laxities linear and above. At any time $t$, CR has an additive overhead on station usage of at most $1+\log (\min\{\max_{c\in C(t)} w_c,\lceil\lceil C(t)\rceil\rceil\}/\min_{c\in C(t)} w_c)$~\footnote{Throughout, $\log$ means $\log_2$ unless otherwise stated.}, for laxities that are powers of $2$. To attain constant amortized reallocation cost, clients are moved to/from the distinguished station only after the number of clients in the system has halved/doubled. 
However, for the reallocation cost function in Equation~\ref{eq:cost}, that is a reallocation cost that depends on the resource requirements of the clients reallocated, CR has an arbitrarily bad reallocation cost ratio, as we show in Lemma~\ref{lemma:CRunbounded}. 

\begin{algorithm}[t]
\SetKwFor{Upon}{upon}{do}{endupon}
\SetKwData{Const}{constant}
\SetKwData{Log}{logarithmic}
\SetKwData{Lin}{linear}
\SetKwFunction{allocate}{allocate}
\SetKwFunction{consolidate}{consolidate}
\SetKwFunction{findLaxityClass}{findLaxityClass}
\SetKwInOut{Input}{input}
\SetKwInOut{Output}{output}
\SetKwFunction{algo}{algo}
\SetKwFunction{proc}{proc}
\SetKwProg{myalg}{Algorithm}{}{}
\SetKwProg{myproc}{Procedure}{}{}
\SetKwProg{myfunc}{Function}{}{}
\DontPrintSemicolon
\myalg{}{
\Upon{arrival or departure of a client $c$}{
	\lIf{arrival}{
			%\allocate{c,\findLaxityClass{c,factor}}
			\allocate{$c,\langle w_{low},w_{high}\rangle$}
	}
	\lElse{
			%\consolidate{c,\findLaxityClass{c,factor}}
			\consolidate{$c,\langle w_{low},w_{high}\rangle$}
		}
	}
}

\myproc{\allocate{$c,\langle w_{low},w_{high}\rangle$}}{
	\For{each depth $i=\lfloor\log w_c\rfloor-\lceil\log w_{low}\rceil$ down to $0$}{
   		\For{each active station $s$ of class $\langle w_{low},w_{high},1/\lfloor\lfloor B/b_c\rfloor\rfloor\rangle$}{
			\If{ there is a leaf $\ell$ available at depth $i$ in the broadcast tree of $s$}{
				allocate to $\ell$ a new subtree with client $c$ assigned at depth $\lfloor\log w_c\rfloor-i-\lceil\log w_{low}\rceil$ of the broadcast subtree\;
				\KwRet\;
			}
		}
   	}
	activate a new station $s$ in class $\langle w_{low},w_{high},1/\lfloor\lfloor B/b_c\rfloor\rfloor\rangle$\;
	choose one of the leaves $\ell$ at depth $0$ of the broadcast subtrees of $s$\;
	allocate to $\ell$ a new subtree with client $c$ assigned at depth $\lfloor\log w_c\rfloor-\lceil\log w_{low}\rceil$ of the broadcast subtree\;
}
\myproc{\consolidate{$c,\langle w_{low},w_{high}\rangle$}}{
	\For{each depth $i=\lfloor\log w_c\rfloor-\lceil\log w_{low}\rceil$ down to $1$}{
		\lIf{ there are two active stations of class $\langle w_{low},w_{high},1/\lfloor\lfloor B/b_c\rfloor\rfloor\rangle$ both with a leaf at depth $i$ available}{
			reallocate sibling subtree of smaller weight\label{realloc1}
		}
		\lElse{\KwRet}
	}
	\tcp{reallocations cleared a whole broadcast subtree}
	\lIf{ there are two active stations of class $\langle w_{low},w_{high},1/\lfloor\lfloor B/b_c\rfloor\rfloor\rangle$ with empty broadcast subtrees}{
		reallocate a subtree from the station with at least one empty subtree to the station with exactly one empty subtree\label{realloc2}
	}
}
\caption{Classified Preemptive Reallocation. $\lfloor\lfloor x\rfloor\rfloor$ is the largest power of $2$ that is not larger than $x$. We represent the transmission schedules with broadcast trees. A node with both children available becomes an available leaf. A station with no client assigned becomes non-active. $\langle w_{low},w_{high}\rangle$ are the boundaries of the class of the input client. Refer to Algorithm~\ref{alg:classifier} %in the Appendix 
for further details on the classification.}
\label{alg:cpr}
\end{algorithm}

\parhead{Classified Preemptive Reallocation}
The negative results in Lemmas~\ref{lemma:PRunbounded},~\ref{lemma:PRexponential}, and~\ref{lemma:CRunbounded} apply to WS. Given that WS is a particular case of SA fixing $b_c=B$ for all clients, the same negative results apply to SA. Thus, should the reallocation cost be maintained low, a new approach is needed.
We present now an online SA protocol (Algorithm~\ref{alg:cpr})  %in Appendix), 
which we call \emph{Classified Preemptive Reallocation} (CPR), that provides guarantees in station usage and reallocation cost. The protocol may be summarized as follows. Clients are classified according to laxity and bandwidth requirements. Upon arrival, a client is allocated to a station within its corresponding class to guarantee a usage excess (with respect to optimal) of at most one station per class plus one station throughout all classes. Upon departure of a client, if necessary to maintain the above-mentioned guarantee, clients are reallocated, but only within the corresponding class. 
%The larger each class is, the higher the reallocation cost. But the smaller each class is, the higher the overhead in station usage. 
The protocol includes three different classifications providing different trade-offs between reallocation cost and station usage. We recreate the idea of broadcast trees, but now we have multiple trees representing the schedule of each station. On one hand, we use broadcast trees with depth bounded by the class laxities. We call them \defn{broadcast subtrees} to reflect that they are only part of a regular broadcast tree. On the other hand, we have the multiplicity yielded by the shared station capacity $B$. An example of broadcast subtrees can be seen in Figure~\ref{fig:alloc}. Further details 
%on the allocation and reallocation mechanisms
follow.

\begin{algorithm}[htb]
\SetKwFor{Upon}{upon}{do}{endupon}
\SetKwData{Const}{constant}
\SetKwData{Log}{logarithmic}
\SetKwData{Lin}{linear}
\SetKwFunction{allocate}{allocate}
\SetKwFunction{consolidate}{consolidate}
\SetKwFunction{findLaxityClass}{findLaxityClass}
\SetKwInOut{Input}{input}
\SetKwInOut{Output}{output}
\SetKwFunction{algo}{algo}
\SetKwFunction{proc}{proc}
\SetKwProg{myalg}{Algorithm}{}{}
\SetKwProg{myproc}{Procedure}{}{}
\SetKwProg{myfunc}{Function}{}{}
\DontPrintSemicolon
\myfunc{\findLaxityClass{c,factor}}{
	\lIf{$1\leq\lfloor\lfloor w_c\rfloor\rfloor<2$ }{\KwRet{$\langle 1,2\rangle$}}
	\lIf{$2\leq\lfloor\lfloor w_c\rfloor\rfloor<4$ }{\KwRet{$\langle 2,4\rangle$}}
	$w \leftarrow 4$\;
	\If{$factor=\Const$}{
		\While(\tcp*[f]{$w_{high}=2w_{low}$}){$\lfloor\lfloor w_c\rfloor\rfloor\geq 2w$}{
			$w \leftarrow 2 w$\;
		}
		\KwRet{$\langle w,2w\rangle$}\;
	}
	\ElseIf{$factor=\Log$}{
		\While(\tcp*[f]{$w_{high}=w_{low}\log_2 w_{low}$}){$\lfloor\lfloor w_c\rfloor\rfloor\geq w\log_2 w$}{
			$w \leftarrow w\log_2 w$\;
		}
		\KwRet{$\langle w,w\log_2 w\rangle$}\;
	}
	\Else(\tcp*[h]{$factor=\Lin$}){
		\While(\tcp*[f]{$w_{high}=w_{low}^2$}){$\lfloor\lfloor w_c\rfloor\rfloor\geq w^2$}{
			$w \leftarrow w^2$\;
		}
		\KwRet{$\langle w,w^2\rangle$}\;
	}
	%\KwRet{$w$}\;
}
\caption{Class Computation. $\lfloor\lfloor x\rfloor\rfloor$ is the largest power of $2$ that is not larger than $x$. The parameter $factor$ indicates how the client classes are defined.}
\label{alg:classifier}
\end{algorithm}

The mechanism to allocate an arriving client can be described as follows. Upon arrival, a client $c$ is classified according to its laxity and bandwidth requirement. Specifically, $c$ is assigned to a class for clients with bandwidth requirement $B/\lfloor\lfloor B/b_c\rfloor\rfloor$ and laxity in $[w_{low},w_{high})$, for some $w_{low}$ and $w_{high}$ that depend on the classification chosen, as shown in Algorithm~\ref{alg:classifier}. %in the Appendix. 
Notice that each station has up to $\lfloor\lfloor B/b_c\rfloor\rfloor\cdot\lceil\lceil w_{low}\rceil\rceil$ subtrees. That is, $\lfloor\lfloor B/b_c\rfloor\rfloor$ ways to share its capacity $B$ and $\lceil\lceil w_{low}\rceil\rceil$ ways to share its schedule (see Figure~\ref{fig:alloc}).
Within its class, we assign $c$ to an available leaf at depth $\lfloor\log w_c\rfloor-\lceil\log w_{low}\rceil$ in any subtree in the forest (see Figure~\ref{fig:alloc1}). 
If there is no such leaf available, we look at smaller depths up in the forest one by one. If we find an available leaf at depth $\lceil\log w_{low}\rceil \leq i<\lfloor\log w_c\rfloor-\lceil\log w_{low}\rceil$, we allocate to that leaf a new subtree with $c$ assigned at depth $\lfloor\log w_c\rfloor-i$ with respect to the root of the broadcast subtree (see Figures~\ref{fig:alloc2} and~\ref{fig:alloc3} ). If no such leaf is available at any depth, a new broadcast subtree $T$ is created with $c$ assigned at depth $\lfloor\log w_c\rfloor-\lceil\log w_{low}\rceil$, and $T$ is assigned to a newly activated station. Refer to Algorithm~\ref{alg:cpr} %in the Appendix 
for further details.

\begin{figure}[t]
\begin{center}
\psfrag{station1}{Station $1$}
\psfrag{station2}{Station $2$}
\subfigure[Arrival of client $i$ with $w_i=8$.]{
\includegraphics[width=0.4\textwidth]{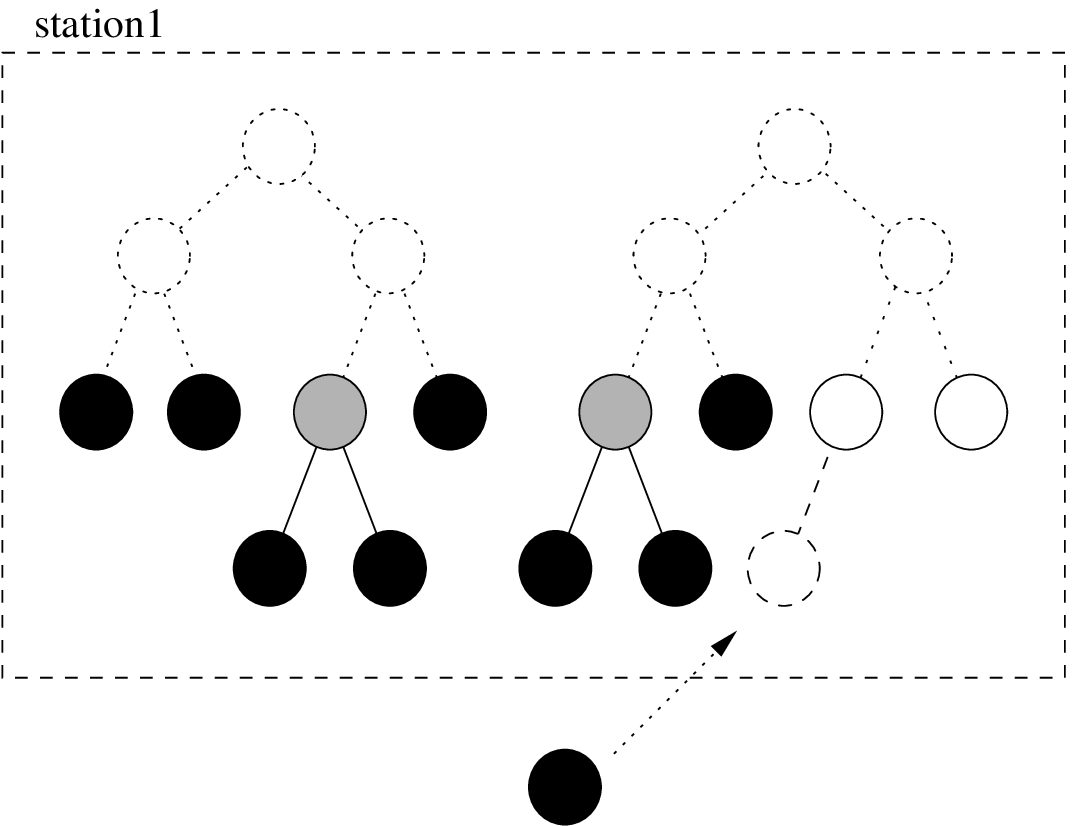}
\label{fig:alloc2}
}
\hfill
\subfigure[Arrival of client $j$ with $w_j=4$.]{
\includegraphics[width=0.4\textwidth]{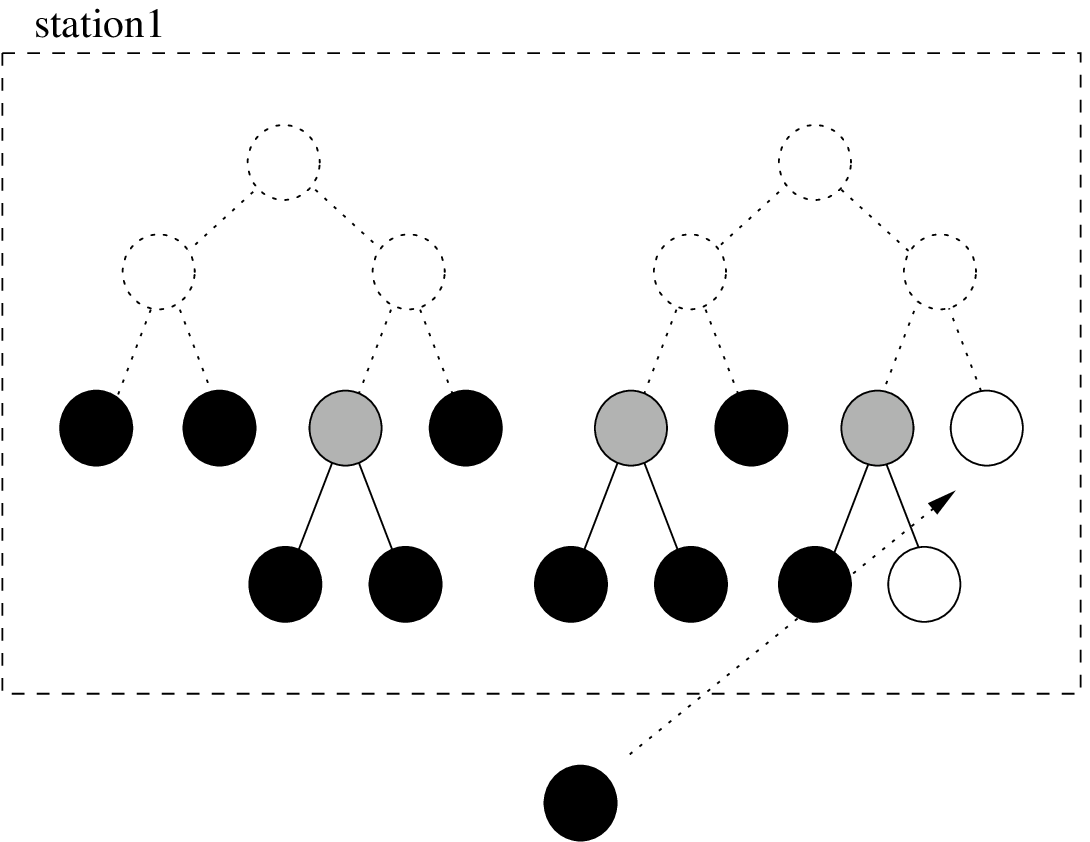}
\label{fig:alloc1}
}
\hfill
\subfigure[Arrival of client $k$ with $w_k=4$.]{
\includegraphics[width=0.80\textwidth]{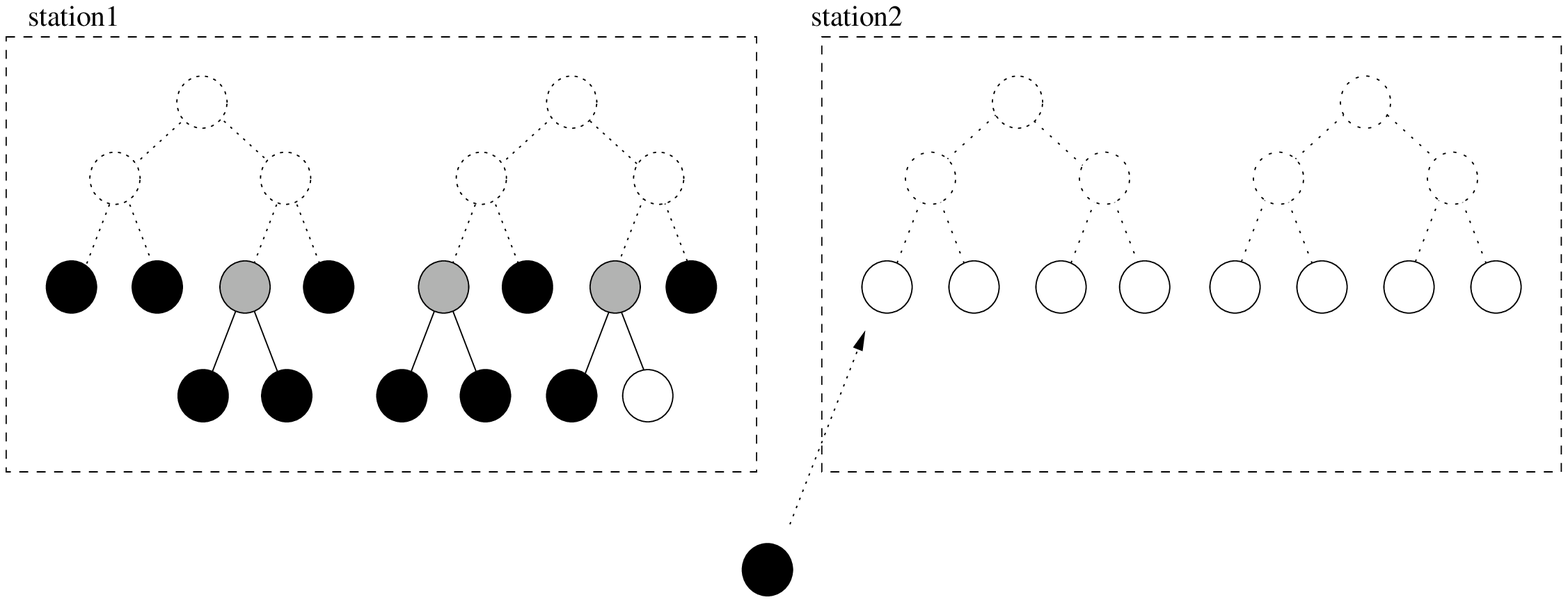}
\label{fig:alloc3}
}
\caption{Illustration of allocation mechanism. Class: laxities $[4,16)$, bandwidth $1/2$. Subtrees are depicted connected to a broadcast tree to reflect their location in the station schedule.} 
\label{fig:alloc}
%\vspace{-20pt}
\end{center}
\end{figure}

The above allocation mechanism maintains the following invariant: (1) there is at most one leaf available at any depth larger than $\lceil\log w_{low}\rceil$ of the forest, and (2) there is at most one station with leaves available at depth $\lceil\log w_{low}\rceil$ (an empty broadcast subtree).
%There might be more than one leaf available at depth $\lceil\log w_{low}\rceil$ (an empty broadcast subtree), but only in one station.
When a client departs, this invariant is re-established through reallocations as follows. 
When a client $c$ departs, if $\lfloor\log w_c\rfloor>\lceil\log w_{low}\rceil$, we check if there was already a leaf $\ell$ available at depth $\lfloor\log w_c\rfloor-\lceil\log w_{low}\rceil$. If there was one, either the sibling of $c$ or the sibling of $\ell$ has to be reallocated to re-establish the invariant. We greedily choose to reallocate whichever sibling has smaller weight of the two (see Figure~\ref{fig:alloc4}). The process does not necessarily stop here because, if $\lfloor\log w_c\rfloor-1>\lceil\log w_{low}\rceil$ and there was a leaf already available at depth $\lfloor\log w_c\rfloor-1-\lceil\log w_{low}\rceil$, together with the newly available leaf at depth $\lfloor\log w_c\rfloor-1-\lceil\log w_{low}\rceil$ due to the reallocation at depth $\lfloor\log w_c\rfloor-\lceil\log w_{low}\rceil$, it yields two leaves available at depth $\lfloor\log w_c\rfloor-1-\lceil\log w_{low}\rceil$. Hence, again one of the sibling subtrees has to be reallocated (see Figure~\ref{fig:alloc5}). This transitive reallocations upwards the forest may continue until a depth where no reallocation is needed or until the depth $\lceil\log w_{low}\rceil+1$ is reached, when the reallocation leaves a broadcast subtree empty. In the latter case, we reallocate a whole broadcast subtree so that only one station has empty subtrees and the invariant is re-established. Refer to Algorithm~\ref{alg:cpr} %in the Appendix 
for further details.

\begin{figure}[t]
\begin{center}
\psfrag{station1}{Station $1$}
\psfrag{station2}{Station $2$}
\psfrag{depart}{depart}
\psfrag{realloc}{reallocate}
\subfigure[Departure of client $j$ with $w_j=4$.]{
\includegraphics[width=0.8\textwidth]{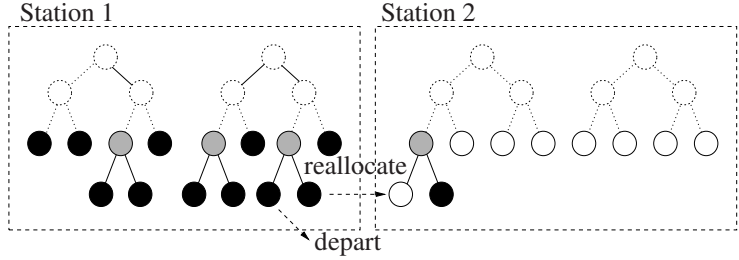}
\label{fig:alloc4}
}
\hfill
\subfigure[Upwards reallocation of sibling with smaller weight.]{
\includegraphics[width=0.8\textwidth]{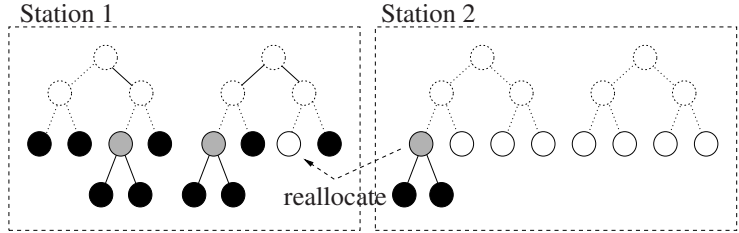}
\label{fig:alloc5}
}
\caption{Illustration of reallocation mechanism. Class: laxities $[4,16)$, bandwidth $1/2$. After the second reallocation Station $2$ is left empty and, hence, deactivated. Subtrees are depicted connected to a broadcast tree to reflect their location in the station schedule.} 
%\vspace{-20pt}
\end{center}
\end{figure}

Notice that when a client is reallocated (even within a station) its laxity may be violated once. Consider for instance the schedule in Figure~\ref{fig:treeSomeClients}. Let $w_a=4$, that is, $a$ is transmitting at its lowest possible frequency. If at the end of time slot $7$ client $b$ departs, at the beginning of time slot $8$ client $a$ will be reallocated to the slot of client $b$, that is, to transmit next in slot $11$. This new schedule violates $w_a$ because the previous slot when $a$ transmitted was $5$. For WS, in~\cite{chan2005temporary} the issue is approached making a client transmit once more within the original schedule. As the authors say, this approach introduces a  transition delay. In their model, there is no impact on station usage because their ratio is against peak load. However, for a ratio against current load such as our model, reserving a slot for a client in more than one station implies an overhead on station usage. Indeed, for any given allocation/reallocation policy, an adversarial input can be shown so that either the laxity is stretched or the station usage is not optimal. Hence, in our model we assume that when a client is reallocated the laxity may be stretched, folding the cost in the reallocation cost.

%%%%%%%%%%%%%%%%%%%%%%%%%%%%%%%%%%%%%%%%%%%%%%%%%%%

\section{Analysis}
\label{section:analysis}

We start with negative results in Lemmas~\ref{lemma:PRunbounded},~\ref{lemma:PRexponential}, and~\ref{lemma:CRunbounded}, which apply to WS, and to SA fixing $b_c=B$ for all clients. 
The proofs are all based on showing an adversarial client set for which the claim holds.

\begin{lemma}
\label{lemma:PRunbounded}
There exists a client arrival/departure schedule such that, in Preemptive Reallocation~\cite{Farach-ColtonLMT14}, the ratio of number of clients reallocated against the number of arrivals plus departures is unbounded.
\end{lemma}
\begin{proof}
Consider the following adversarial client arrival/departure schedule divided in rounds.
In the first round, $2$ clients of laxity $2$ arrive.
Then, for each round $r=2,3,4,\dots$, two clients of laxity $2^r$ arrive and, after these clients have been allocated, a client of laxity $2^{r-1}$ departs.
Figure~\ref{fig:PRunbounded} shows the status of the forest right before each departure. 

\begin{figure}[htbp]
\begin{center}
\psfrag{dep}{depart}
\psfrag{realloc}{reallocate}
\subfigure[First departure.]{
\includegraphics[width=0.23\textwidth]{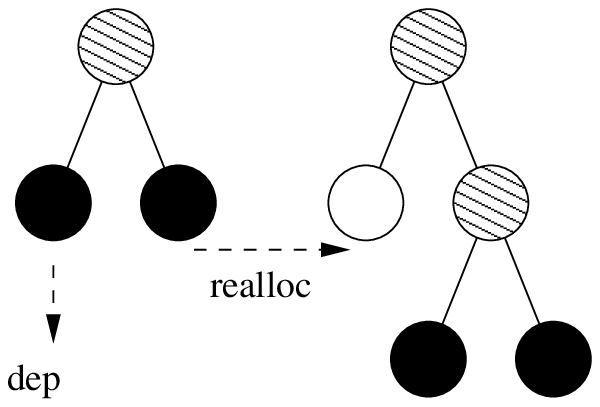}
\label{fig:1Dep}
}
\hfill
\subfigure[Second departure.]{
\includegraphics[width=0.23\textwidth]{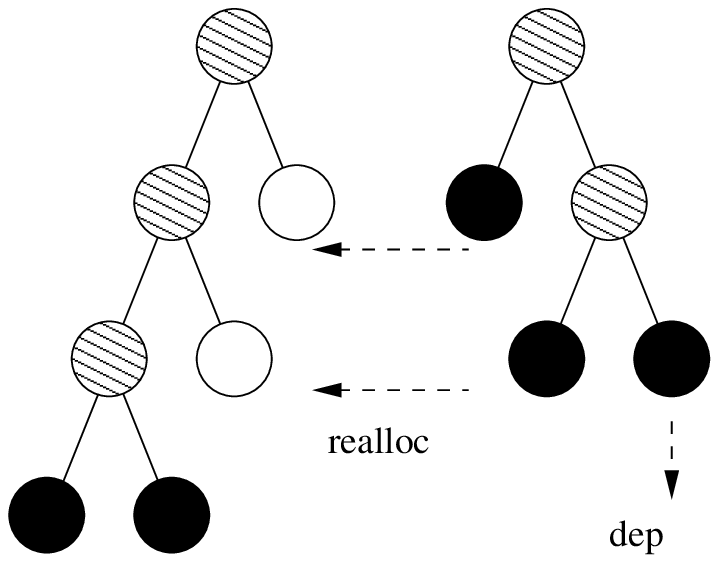}
\label{fig:2Dep}
}
\hfill
\subfigure[Third departure.]{
\includegraphics[width=0.23\textwidth]{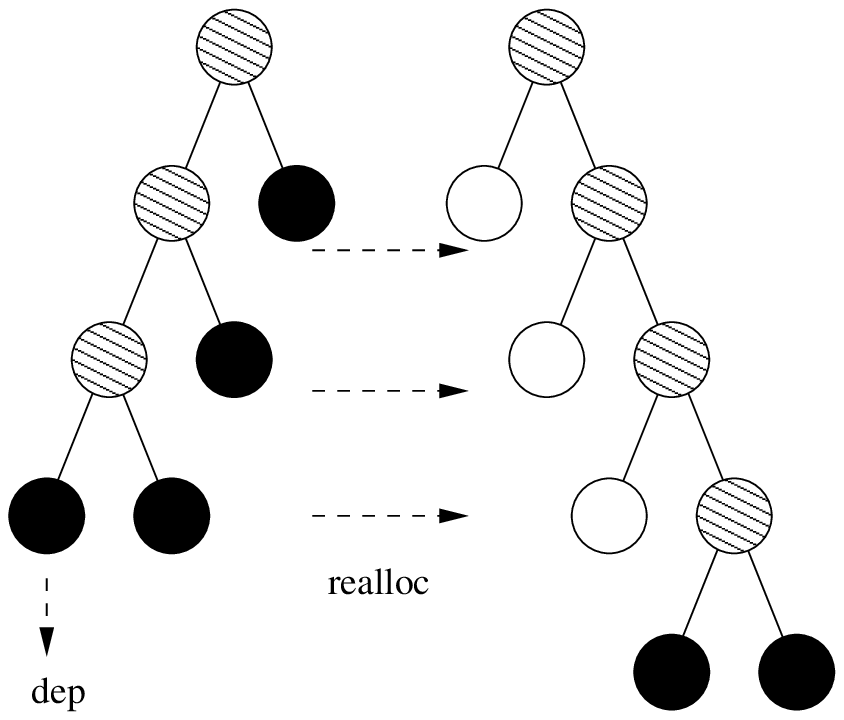}
\label{fig:3Dep}
}
\hfill
\subfigure[$i$th departure.]{
\includegraphics[width=0.23\textwidth]{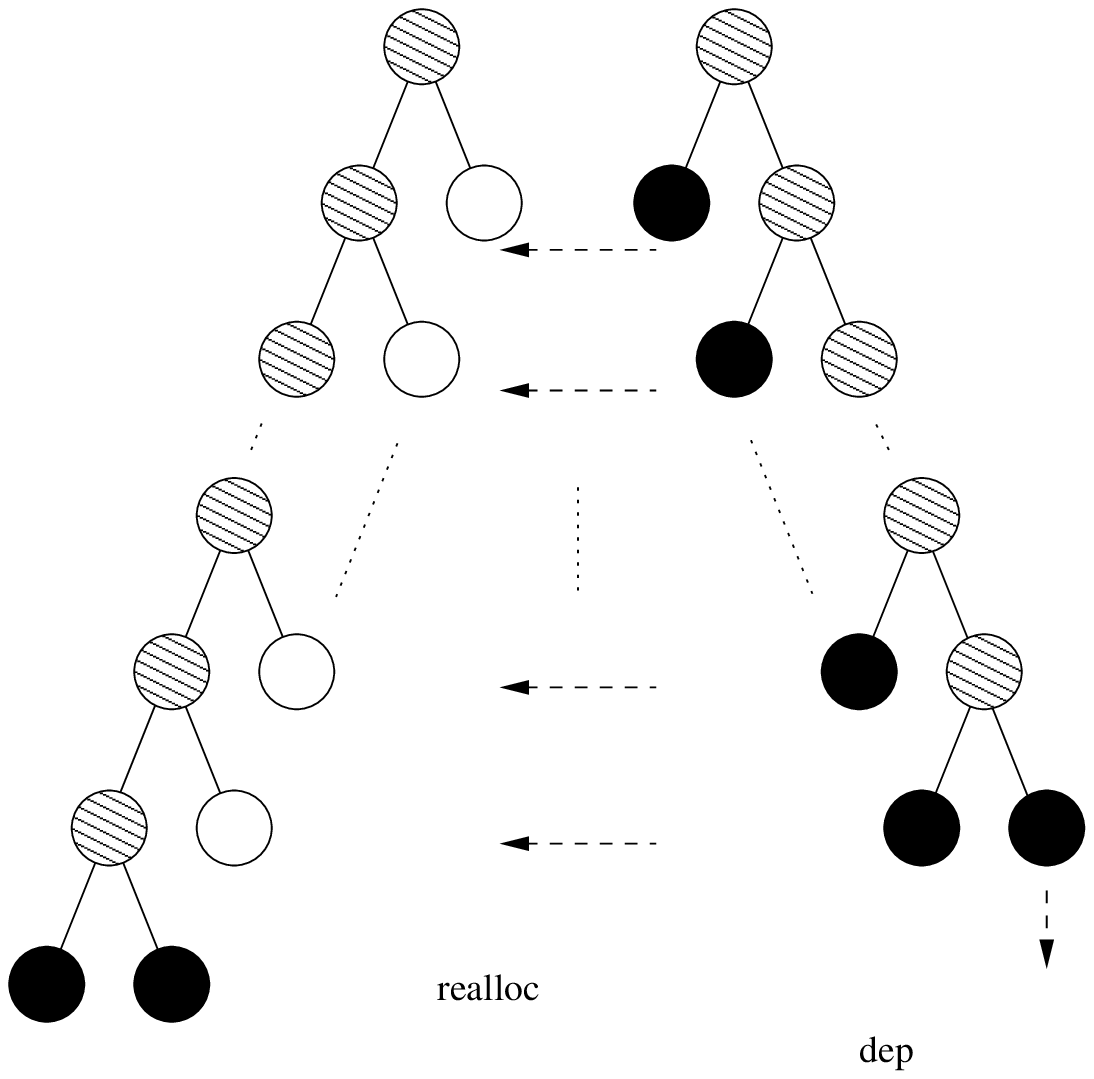}
\label{fig:manyDep}
}
\caption{Illustration of Lemma~\ref{lemma:PRunbounded}.} 
\label{fig:PRunbounded}
\end{center}
\end{figure}

To compute the reallocation cost, consider any round $r\geq 2$. After the departure, two leaves are left available at depth $r$ of the forest. For example, refer to Figure~\ref{fig:3Dep} depicting round $4$. After the client at depth $4$ departs, two leaves are left available at that depth. To restore the invariant, PR reallocates the sibling subtrees of the available leaves, so that they are assigned to the same parent node. In doing so, now two leaves are left available at depth $r-1$ of the forest. Because PR reallocates the subtree with less clients assigned, similar reallocations are repeated transitively up through the forest until one of the trees is left empty. (Refer to Figure~\ref{fig:3Dep}.) Then, the number of reallocated clients in round $r$ is $r$, whereas the number of arriving or departing clients in each round is always $3$. Given that the number of rounds is infinite, the overall reallocation cost ratio is unbounded.
\qed
\end{proof}

\begin{lemma}
\label{lemma:PRexponential}
For Preemptive Reallocation~\cite{Farach-ColtonLMT14}, modified so that the sibling subtree of smaller {\bf weight} is reallocated to restore the invariant, rather than the subtree with less clients, the following holds.
For any $d>0$, there exists a client arrival/departure schedule such that 
it is $\max_{t:\mathcal{R}(t)>0} \mathcal{R}(t)/\mathcal{D}(t) \geq \rho(2^d-1)^2/2^d$.
\end{lemma}
\begin{proof}
%In the following, it is $b_c=B$ for all clients $c$.
Given $d>0$, consider the following adversarial client arrival/departure schedule divided in phases.
First a client of laxity $2^d$ arrives.
After this client was assigned, a sequence of clients arrive one by one so that a new client arrives only after the previous client was assigned. The sequence of laxities of those clients is the following. 
\begin{align*}
&2^{d+1},2^{d+2},\dots,2^{2d-1},2^{2d},\\
2^{d},&2^{d+1},\dots,2^{2d-2},2^{2d-1}, \\
2^{d-1},2^{d},&2^{d+1},\dots,2^{2d-2}, \\
\dots\\
2^{2},2^{3},\dots,2^{d-1},2^{d},&2^{d+1}.
\end{align*} 
Then, another client of laxity $2^d$ arrives.
Figure~\ref{fig:expPR3r0} illustrates the assignment of clients by PR for $d=3$.
Finally, after all clients have been assigned, the client that arrived first departs.
No other client arrives or departs afterwards.
The client departure leaves two leaves available at depth $d$.
Then, the sibling subtree of smaller weight is reallocated (refer to Figure~\ref{fig:expPR3r0}).
In turn, this reallocation leaves two leaves available at depth $d-1$, which triggers the reallocation of the sibling subtree of smaller weight  (refer to Figure~\ref{fig:expPR3r1}).
These transitive reallocations continue upwards the tree depth-by-depth up to depth $2$ (refer to Figure~\ref{fig:expPR3r2}), when the last reallocation leaves one of the trees empty  (refer to Figure~\ref{fig:expPR3r3}).

\begin{figure}[htbp]
\begin{center}
\psfrag{dep}{depart}
\psfrag{realloc}{reallocate}
\subfigure[Before first reallocation.]{
\includegraphics[width=0.4\textwidth]{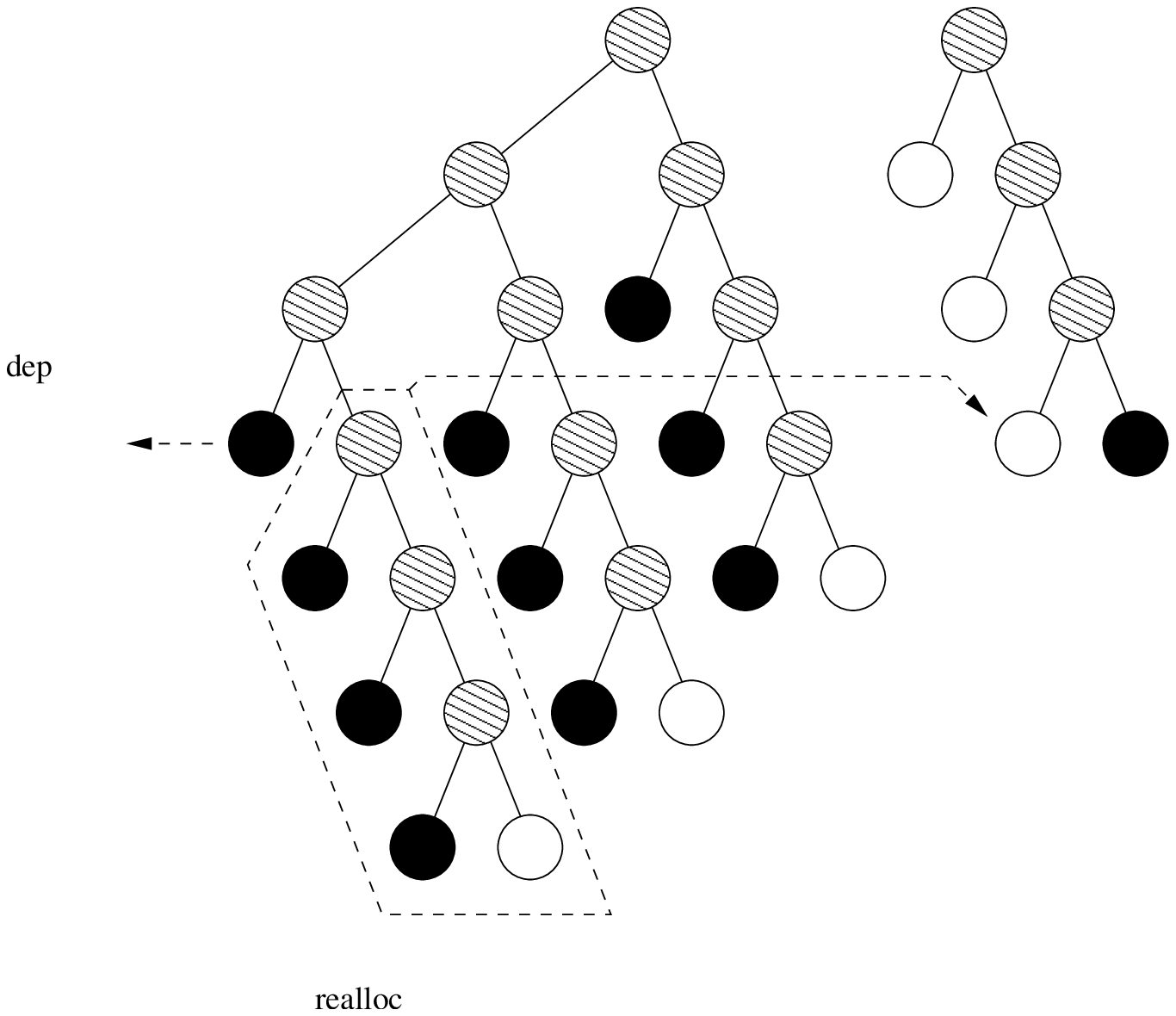}
\label{fig:expPR3r0}
}
\hfill
\subfigure[Before second reallocation.]{
\includegraphics[width=0.4\textwidth]{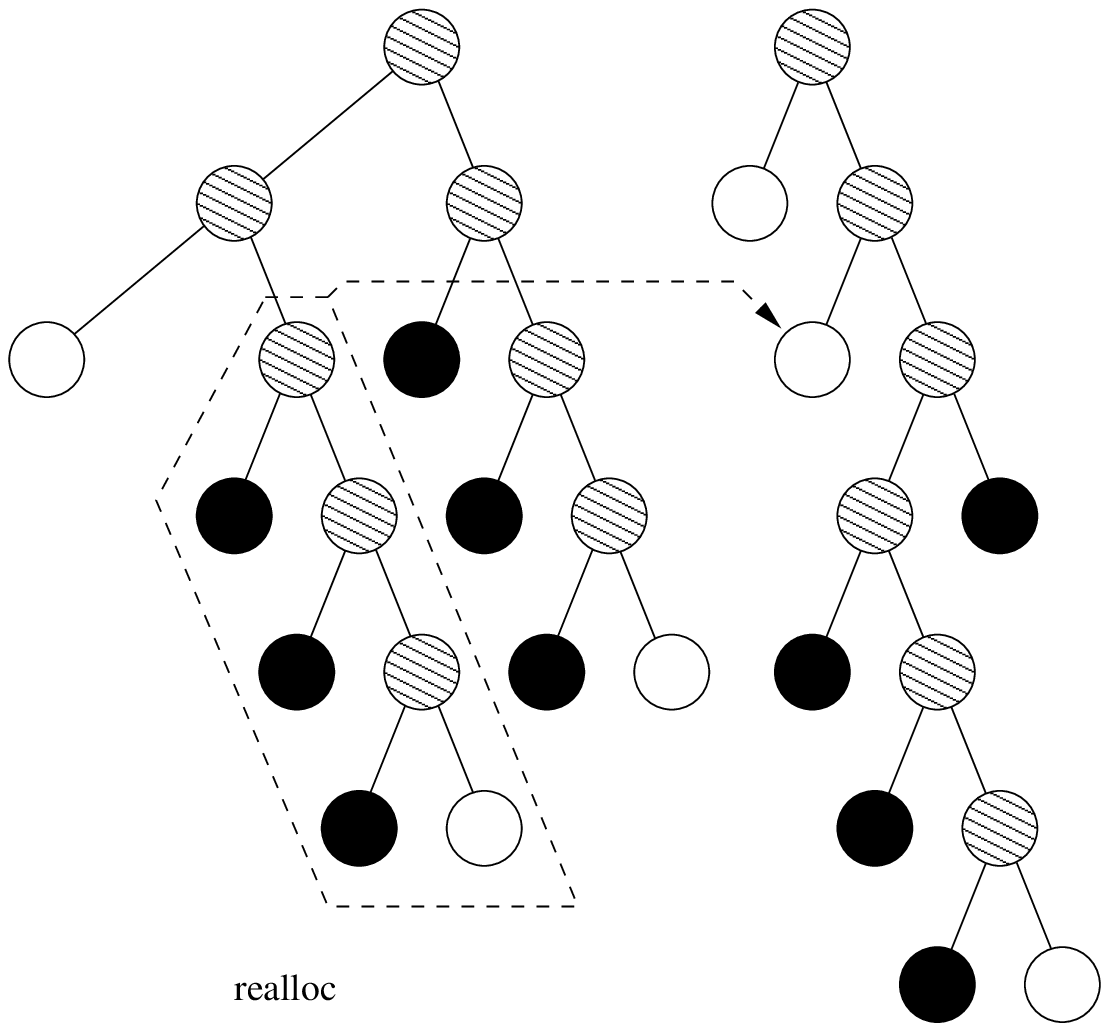}
\label{fig:expPR3r1}
}
\hfill
\subfigure[Before third reallocation.]{
\includegraphics[width=0.4\textwidth]{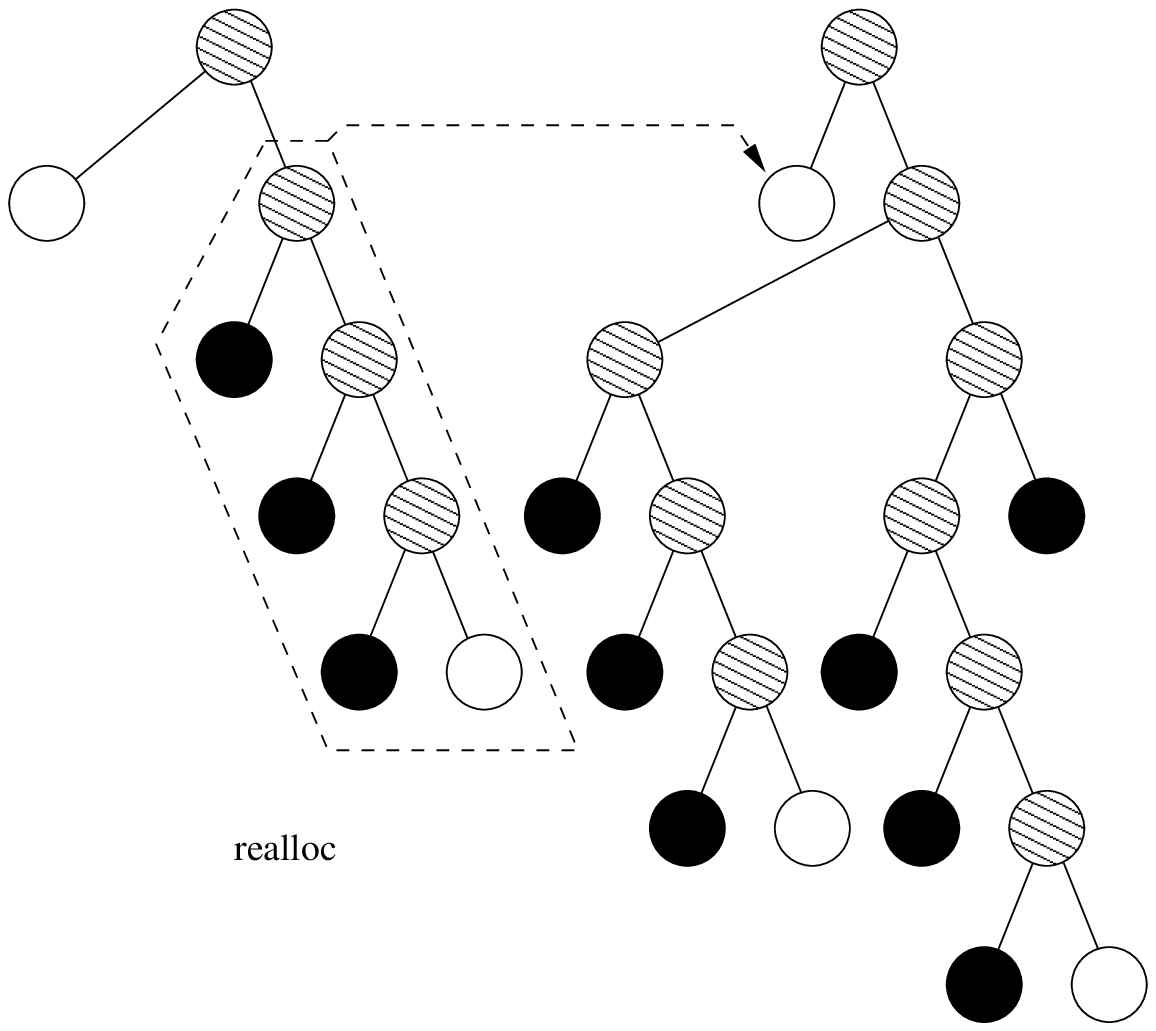}
\label{fig:expPR3r2}
}
\hfill
\subfigure[Final assignment.]{
\includegraphics[width=0.4\textwidth]{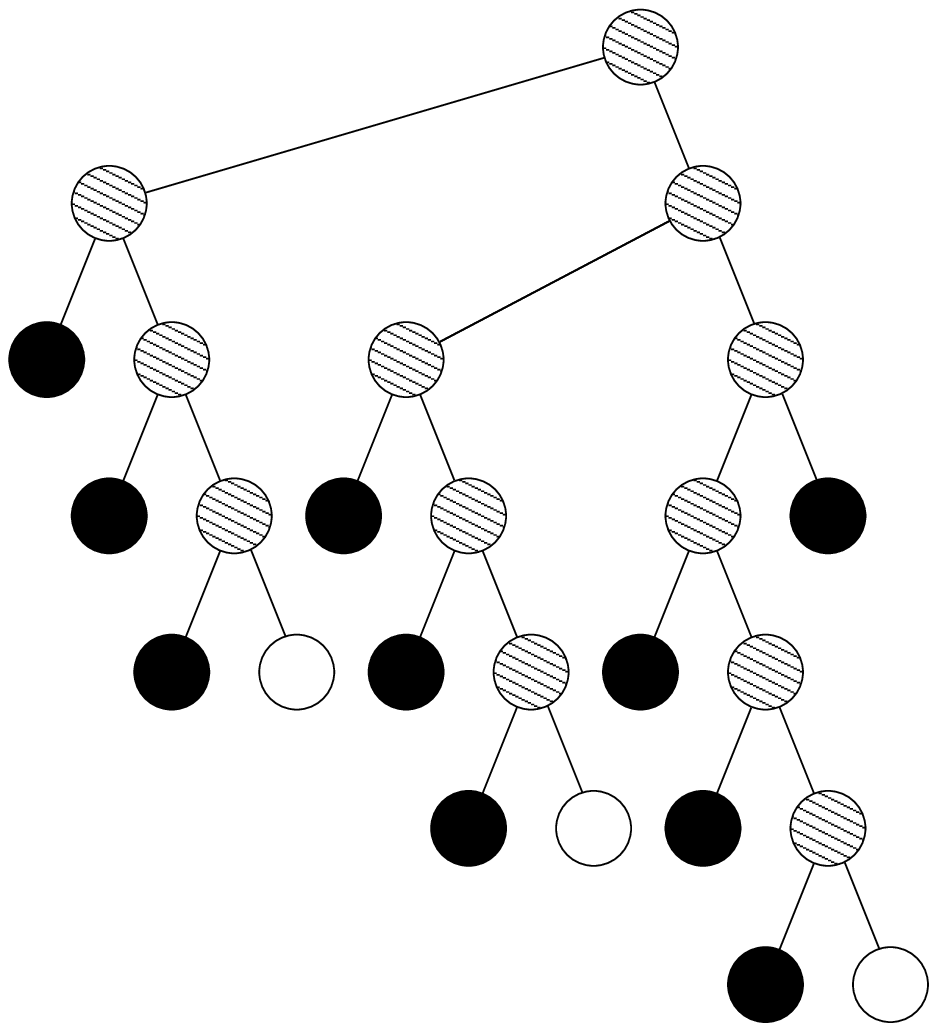}
\label{fig:expPR3r3}
}
\caption{Illustration of Lemma~\ref{lemma:PRexponential}.} 
\label{fig:expPR3}
\end{center}
\end{figure}

Then, at the time slot $t$ when all clients have been reallocated, we have
\begin{align*}
\frac{\mathcal{R}(t)}{\mathcal{D}(t)} 
&= \frac{\rho\sum_{c\in R(t)}1/w_c}{1/2^d}\\
&= \frac{\rho\sum_{i=2}^{d+1}\sum_{j=0}^{d-1}1/2^{i+j}}{1/2^d}\\
&= \frac{\rho(2^d-1)^2}{2^d}.
\end{align*}

\vspace{-5ex}
\qed
\end{proof}

\begin{lemma}
\label{lemma:CRunbounded}
For any integer $x>0$ and any $w\geq 2^{x+5}$ arbitrarily big such that $w$ is a power of 2, there exists a client arrival/departure schedule such that, in Classified Reallocation~\cite{Farach-ColtonLMT14}, we have $\max_{t:\mathcal{R}(t)>0} \mathcal{R}(t)/\mathcal{D}(t) \geq \frac{\rho/4}{7\cdot2^x}w$.
\end{lemma}
\begin{proof}
% \pru{[channel $\rightarrow$ station?]}
We use the terminology ``channel'' in~\cite{Farach-ColtonLMT14} in this proof.
The thresholds to reallocate from/to the big channel in CR are the following~\cite{Farach-ColtonLMT14}. For any time $t$, if a client $c$ allocated to the big channel has laxity $w_c<\lceil\lceil |C(t)| \rceil\rceil$, $c$ is reallocated to other channel according to $w_c$, call it $w_c$-channel. On the other hand, if at any time $t$ a client $c$ that is \emph{not} allocated to the big channel has laxity $w_c>2\lceil\lceil |C(t)| \rceil\rceil$, then $c$ is reallocated to the big channel.

%In the following, it is $b_c=B$ for all clients $c$.
Consider an adversarial scenario where the system has $2^x$ clients with laxity $2^{x+2}$ and $7 \cdot 2^x$ clients with laxity $w$, where $w$ is a power of $2$ such that $w\geq 2^{x+5}$. (The order in which these clients have arrived is irrelevant.) Because the total number of clients is $2^{x+3}$, the clients with laxity $w\geq 2^{x+5}>2\cdot 2^{x+3}$ are allocated to the big channel, whereas the clients with laxity $2^{x+2}<2^{x+3}$ are allocated to a ($2^{x+2}$)-channel.
After these clients have been allocated, adversarially, all the clients with laxity $w$ depart. Because the new number of clients in the system is now $2^x$, the remaining clients, all with laxity $2^{x+2}>2\cdot 2^x$ have to be reallocated to the big channel. Then, at time $t$ after reallocation, the following holds.
\begin{align*}
\frac{\mathcal{R}(t)}{\mathcal{D}(t)} 
&= \frac{\rho\sum_{c\in R(t)}1/w_c}{\sum_{c\in D(t)}1/w_c}\\
&= \frac{\rho/4}{7\cdot2^x}w.
\end{align*}
\qed
\end{proof}

The above lemmas show that the application of previous WS reallocation algorithms to SA is not feasible. 
The following theorem gives guarantees on station usage and reallocation cost for CPR.
The proof %, left to the Appendix for brevity, 
starts by analyzing CPS to show that the invariant is re-established after each arrival or departure. Then, competitiveness on station usage is derived from the invariant properties. 
Finally, to bound $\beta$, a worst case scenario minimizing the weight of departed clients and maximizing the reallocated weight is shown.

\begin{theorem}
\label{thm:main}
At any time slot $t$, CPR achieves an $(\alpha,\beta)$-\name as follows.
\begin{align*}
\alpha &= \max_t \frac{4(1+ \Gamma(ALG,t) + S(OPT,t))}{S(OPT,t)}\\
%\max_t \{4+ 2\Gamma(ALG,t)\}\\
\beta &= \max_t\rho(2\lfloor\lfloor w_{high_{\max}}(t)\rfloor\rfloor/\lceil\lceil w_{low_{\max}}(t)\rceil\rceil - 1).
\end{align*}
Where $\Gamma(ALG,t)$ is the number of classes used by CPR at time $t$, and $w_{high_{\max}}(t)$ and $w_{low_{\max}}(t)$ are the maximum upper and lower limits of a class at time $t$.
\end{theorem} 
\begin{proof}
We start by showing that the invariant in Algorithm~\ref{alg:cpr} is preserved. Recall that the invariant is the following. At any time slot $t$ and for any class of clients $\langle w_{low},w_{high},x\rangle$, there is at most one leaf available at any depth larger than $\lceil\log w_{low}\rceil$ of the forest. There might be more than one leaf available at depth $\lceil\log w_{low}\rceil$ (an empty broadcast subtree), but only in one station in the class. 

The arrival of clients does not change the invariant, but the departure of a client $c$ at a given depth $i>\lceil\log w_{low}\rceil$ may change the number of leaves available at depth $i$. If there was no leaf available at depth $i$ before the departure, the number of available leaves at depth $i$ is at most one after departure and the invariant is preserved. 
If, on the other hand, there was a leaf $\ell$ available at depth $i$, either the sibling of $c$ or the sibling of $\ell$ will be reallocated in Line~\ref{realloc1} of the algorithm. This reallocation leaves two sibling leaves available at depth $i$, which combined yield a leaf available at depth $i-1$. The same argument applies transitively upwards the tree. If the invariant is re-established before reaching depth $\lceil\log w_{low}\rceil$, we are done. If on the other hand a broadcast subtree is emptied, the invariant is re-established (if necessary) reallocating a whole broadcast subtree in Line~\ref{realloc2}. Notice that reallocating one subtree is enough to re-establish the invariant, since before the departure there was (at most) one station with empty subtrees, and the departure (possibly followed by reallocations) may empty only one subtree.

To bound $\alpha$, we observe that the invariant above guarantees that there is at most one station per class with empty broadcast subtrees. For the stations with non-empty subtrees, aggregating the at most one available leaf at each depth larger than $0$ (and smaller than $\lfloor\log w_{high}\rfloor$) of each forest, we have an additional available space of at most one station, throughout all classes. So, the overhead in station usage is the number classes plus one. Additionally, we have to take into account that clients are scheduled to transmission periods that are powers of $2$, and with a bandwidth that is a power of $2$ fraction of the capacity $B$, which introduces a multiplicative factor in station usage of at most $4$. Thus, we have
\begin{align*}
\max_{t} \frac{S(ALG,t)}{S(OPT,t)} 
&\leq \max_t \frac{4(1+ \Gamma(ALG,t) + S(OPT,t))}{S(OPT,t)} \\
%&= \max_t 2+\frac{2(1+ \Gamma(ALG,t))}{S(OPT,t)} \\
%&\leq \max_t \{4+ 2\Gamma(ALG,t)\}, \textrm{ because $S(OPT,t)\geq 1$.} 
\end{align*}

To bound $\beta$, we compute the maximum weight of clients reallocated upon a departure. We notice that, for any class of clients $\langle w_{low},w_{high},x\rangle$, in the worst case a departure at depth $\lfloor\log w_{high}\rfloor$ triggers transitive reallocations upwards up to depth $\lceil\log w_{low}\rceil-1$ in the forest, followed by a reallocation of a whole broadcast subtree of weight at most $1/\lceil\lceil w_{low}\rceil\rceil$. The aggregated weight of all those reallocations is then $1/\lceil\lceil w_{low}\rceil\rceil + 1/(2\lceil\lceil w_{low}\rceil\rceil)+ 1/(4\lceil\lceil w_{low}\rceil\rceil)+\dots +1/\lfloor\lfloor w_{high}\rfloor\rfloor = 2/\lceil\lceil w_{low}\rceil\rceil - 1/\lfloor\lfloor w_{high}\rfloor\rfloor$. Replacing, we obtain
\begin{align*}
\max_{t:\mathcal{R}(ALG,t)>0} \frac{\mathcal{R}(ALG,t)}{\mathcal{D}(ALG,t)} 
&\leq \max_{w_{low},w_{high}}\frac{\rho(2/\lceil\lceil w_{low}\rceil\rceil - 1/\lfloor\lfloor w_{high}\rfloor\rfloor)}{1/\lfloor\lfloor w_{high}\rfloor\rfloor}\\
%&= \max_{w_{low},w_{high}}\rho(2\lfloor\lfloor w_{high}\rfloor\rfloor/\lceil\lceil w_{low}\rceil\rceil - 1)\\
&\leq \max_t\rho(2\lfloor\lfloor w_{high_{\max}}(t)\rfloor\rfloor/\lceil\lceil w_{low_{\max}}(t)\rceil\rceil - 1).
\end{align*}
\qed
\end{proof}

Instantiating Theorem~\ref{thm:main} in the classification factors of Algorithm~\ref{alg:classifier}, we obtain bounds for all three algorithms, shown in 
Corollary~\ref{cor:gral}. %in the Appendix.
\begin{corollary}
\label{cor:gral}
At any time slot $t$, CPR achieves an $(\alpha,\beta)$-\name as follows.
\begin{enumerate}
\item
\pru{\textbf{Constant factor.}}
If the client classification boundaries are
$[w_i,w_{i+1})$, where $w_1=1$, and $w_i=2w_{i-1}$, for any $i>1$, then
\begin{align*}
\alpha &= 4\left(1+\frac{1+\left(1+\log \frac{\lceil\lceil B/b_{\min}(t)\rceil\rceil}{\lceil\lceil B/b_{\max}(t)\rceil\rceil}\right)\left(1+\log \frac{\lfloor\lfloor w_{\max}(t)\rfloor\rfloor}{\lfloor\lfloor w_{\min}(t)\rfloor\rfloor}\right)}{H(C(t))}\right)\\
\beta &= 3\rho.
\end{align*}
\item
\pru{\textbf{Logarithm factor.}}
If the client classification boundaries are
$[w_i,w_{i+1})$, where $w_1=1, w_2=2, w_3=4$, and $w_i=w_{i-1} \log w_{i-1}$, for any $i>3$, then
\begin{align*}
\alpha &= 4\left(1+\frac{1+\left(1+\log \frac{\lceil\lceil B/b_{\min}(t)\rceil\rceil}{\lceil\lceil B/b_{\max}(t)\rceil\rceil}\right)\left(1+\frac{\log\lfloor\lfloor w_{\max}(t)\rfloor\rfloor}{\log\log\max\{4,\lfloor\lfloor w_{\min}(t)\rfloor\rfloor\}}\right)}{H(C(t))}\right)\\
\beta &= \rho(2\log w_{\max}(t)-1).
\end{align*}
\item
\pru{\textbf{Linear factor.}}
If the client classification boundaries are
$[w_i,w_{i+1})$, where $w_1=1, w_2=2$, and $w_i=w_{i-1}^2$, for any $i>2$, then
\begin{align*}
\alpha &= 4\left(1+\frac{1+\left(1+\log \frac{\lceil\lceil B/b_{\min}(t)\rceil\rceil}{\lceil\lceil B/b_{\max}(t)\rceil\rceil}\right)\left(1+\log \frac{\log\max\{2,\lfloor\lfloor w_{\max}(t)\rfloor\rfloor\}}{\log\max\{2,\lfloor\lfloor w_{\min}(t)\rfloor\rfloor\}}\right)}{H(C(t))}\right)\\
\beta &= \rho\left(2\sqrt{w_{\max}(t)}-1\right).
\end{align*}
\end{enumerate}
Where $H(C(t))=\lceil\sum_{c\in C(t)}1/w_c\rceil$,
$w_{\max}(t) = \max_{c\in C(t)} w_c$,
$w_{\min}(t)= \min_{c\in C(t)} w_c$,
$b_{\max}(t) = \max_{c\in C(t)} b_c$,
and $b_{\min}(t)= \min_{c\in C(t)} b_c$.
\end{corollary} 

\begin{proof}
Using that $S(OPT,t)\geq H(C(t))$,
and bounding the values of $\max_{t}\Gamma(ALG,t)$ and $\max_t \lfloor\lfloor w_{high_{\max}}(t)\rfloor\rfloor/\lceil\lceil w_{low_{\max}}(t)\rceil\rceil$ in Theorem~\ref{thm:main}, the claim follows.
\qed
\end{proof}

\pru{We note that the choice of classification factor gives a trade-off on the performance on station usage and reallocation cost,
i.e., the station usage improves as we move from constant to logarithm to linear factor while
the reallocation cost improves as we move from linear to logarithm to constant factor.
We comment that the logarithm classification gives good performance for both measurement.
}

%To provide intuition and comparison for the simulations, 
To provide intuition,
we instantiate Corollary~\ref{cor:gral} on a setting where all laxities are powers of $2$ and all bandwidth requirements are the full capacity of a station, as follows.

\begin{corollary}
\label{cor:part}
For a set of clients $C$ such that, for all $c\in C$, it is $b_c=B$ and $w_c=2^i$ for some $i\geq 0$, and for all $t$ it is $w_{\max}(t)>w_{\min}(t)\geq 4$, the following holds. At any time slot $t$, CPR achieves an $(\alpha,\beta)$-\name as follows.
\begin{enumerate}
\item
If the client classification boundaries are
$[w_i,w_{i+1})$, where $w_1=1$, and $w_i=2w_{i-1}$, for any $i>1$, then
\begin{align*}
\alpha &= 1+\left(2+\log (w_{\max}(t)/w_{\min}(t))\right)/H(C(t))\\
\beta &= 3\rho.
\end{align*}
\item
If the client classification boundaries are
$[w_i,w_{i+1})$, where $w_1=1, w_2=2, w_3=4$, and $w_i=w_{i-1} \log w_{i-1}$, for any $i>3$, then
\begin{align*}
\alpha &= 1+\left(2+\log w_{\max}(t)/\log\log w_{\min}(t)\right)/H(C(t))\\
\beta &= \rho(2\log w_{\max}(t)-1).
\end{align*}
\item
If the client classification boundaries are
$[w_i,w_{i+1})$, where $w_1=1, w_2=2$, and $w_i=w_{i-1}^2$, for any $i>2$, then
\begin{align*}
\alpha &= 1+\left(2+\log (\log w_{\max}(t)/\log w_{\min}(t))\right)/H(C(t))\\
\beta &= \rho\left(2\sqrt{w_{\max}(t)}-1\right).
\end{align*}
\end{enumerate}
Where  $H(C(t))=\lceil\sum_{c\in C(t)}1/w_c\rceil$,
$w_{\max}(t) = \max_{c\in C(t)} w_c$,
$w_{\min}(t)= \min_{c\in C(t)} w_c$,
$b_{\max}(t) = \max_{c\in C(t)} b_c$,
and $b_{\min}(t)= \min_{c\in C(t)} b_c$.
\end{corollary}

\section{Simulations}
In this section, we present the main results of our experimental simulations of the CPR algorithm. 
%We highlight here that the classification factor (logarithmic) that balances station usage and reallocation cost was found through experimentation with various functions. 
We highlight here that the classification factor (logarithmic) that maintains simultaneously station usage and reallocation cost below the maximum observed was found through experimentation with various functions. 
For the specific cases presented (constant, logarithmic, and linear factors) we have focused on a scenario 
%where $\forall c\in C, b_c=B$ and $w_c=2^i, i\geq 0$ (as in Corollary~\ref{cor:part}). 
%Simulations for arbitrary bandwidths and laxities are left for future work.
where $\forall c\in C, b_c=1/2^i$, and $w_c=2^j,$ where $i,j\geq 0$ and $B$ was normalized to $1$.
For all the evaluations the reallocation cost of each client $c$ has been set to the inverse of its laxity $1/w_c$. That is, $\rho=1$, since the scaling factor $\rho$ is also a multiplicative factor in our reallocation metric and, hence, does not provide additional information about the performance of our protocols in terms of reallocation.

Our theoretical bounds on performance apply to worst-case scenarios. Hence, the purpose of these simulations is to complement those bounds evaluating how much better (if anything) our protocol behaves in practice for average cases. Given that the main feature of the protocol is to allocate (and reallocate) ``efficiently'', we aim to stress such feature considering 
inputs that entail extremal cases of arrivals. That is, smooth distributions of arrivals as well as 
batched arrivals. The set of inputs chosen are representative of those cases. Moreover, they are also
the customary choices in experimental evaluation for other problems such as job scheduling, packet routing, etc.
Other reallocation algorithms were not simulated since, to the best of our knowledge, this is the first time that restrictions on laxity and bandwidth under a reallocation cost proportional to resources requested have been considered.

%To evaluate thoroughly the performance of our protocol, w
We have produced various sets of clients (recall that each client is characterized by a time of arrival, a time of departure, a bandwidth, and a laxity). 
The laxity of each client was chosen independently at random from $\{1,2,4,\dots,w_{\max}\}$, for each $w_{\max}=1024, 4096,$ and $16384$. 
We evaluated three distributions over that range: uniform, biased towards small laxities, and biased towards large laxities. 
Biased \pru{means} probability $0.7$ of choosing from one half of the range (lower or higher), and then uniform probability within the half chosen.
%for each client $c$, $w_c=1$ with probability $1/1024$, or $w_c=2^i$ with probability $2^i/2^{11}$, for $1\leq i\leq10$. 
The bandwidth of each client $c$ was chosen at random as $b_c=1/2^i$ with probability $1/2^i$ for each $i=1,2,\dots$
For each of $n=4000, 8000,$ and $16000$ clients, time was discretized in $2n$ slots.

\begin{figure}[h]
\begin{center}
\resizebox{.8\textwidth}{!}{\input{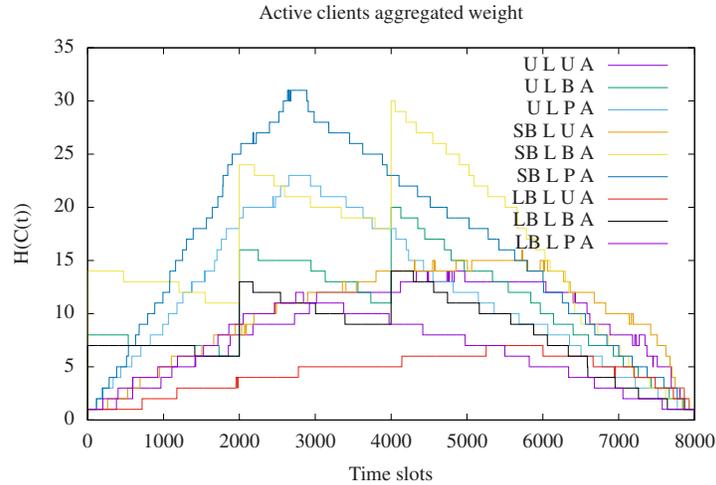}}
\caption{Cumulative inverse laxity ($H(C(t))$) vs. time for $n=4000$ and $w_{\max}=1024$.
Key: L: laxity, A: arrival, U: uniform, B: batched, P: Poisson, SB: small-biased, LB: large-biased.
}
\label{fig:H}
\end{center}
\end{figure}

\begin{figure}[h]
%\begin{center}
%\subfigure{\resizebox{.6\textwidth}{!}{\input{epslatex_betastatsoverdepartures.tex}}}
\subfigure{\resizebox{.85\textwidth}{!}{\input{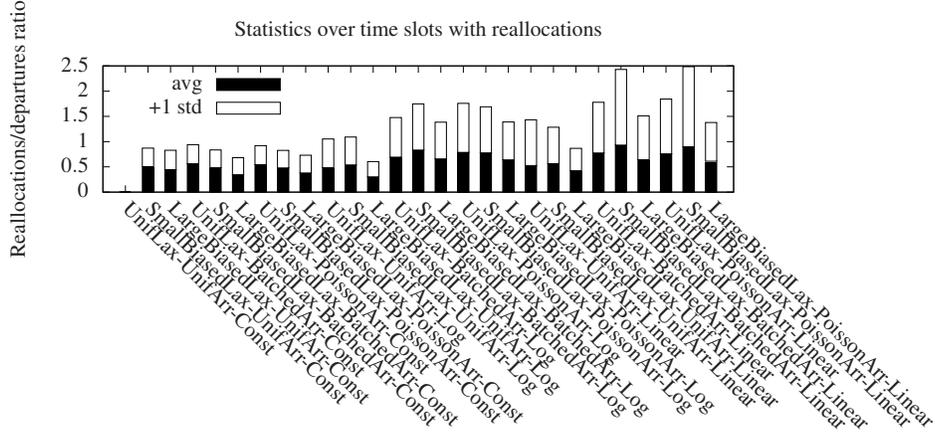}}}
\vspace{-30pt}
\caption{Reallocation/Departure ratio statistics for different classification factors, laxity distributions, and arrival distributions,  for $n=4000$ and $w_{\max}=1024$.}
\label{fig:betabars}
%\end{center}
\end{figure}

\begin{figure}[h]
\begin{center}
\resizebox{.8\textwidth}{!}{\input{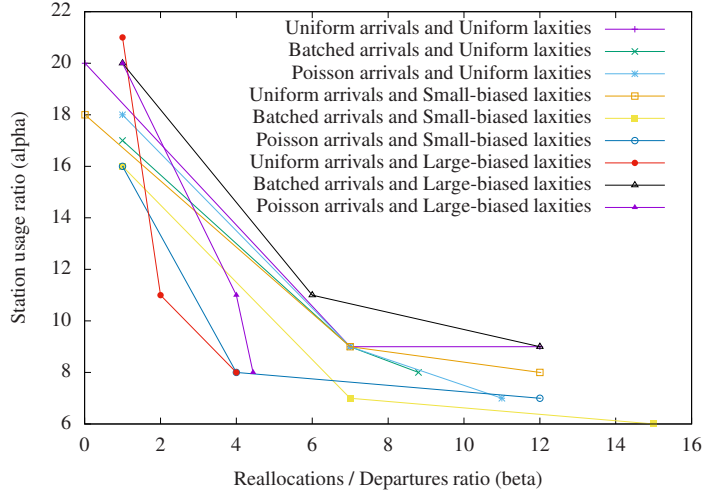}}
\caption{Worst case $\alpha$ vs. $\beta$. $n=4000$, $w_{\max}=1024$, $w_{\min}=1$, $\rho=1$. Each colored line corresponds to the 3 data points obtained for each input. The leftmost point corresponds to the Constant factor algorithm, the middle point corresponds to the Logarithmic factor, and the rightmost point corresponds to the Linear factor.}
\label{fig:alphavsbeta}
%\vspace{-20pt}
\end{center}
\end{figure}

\begin{figure}[h]
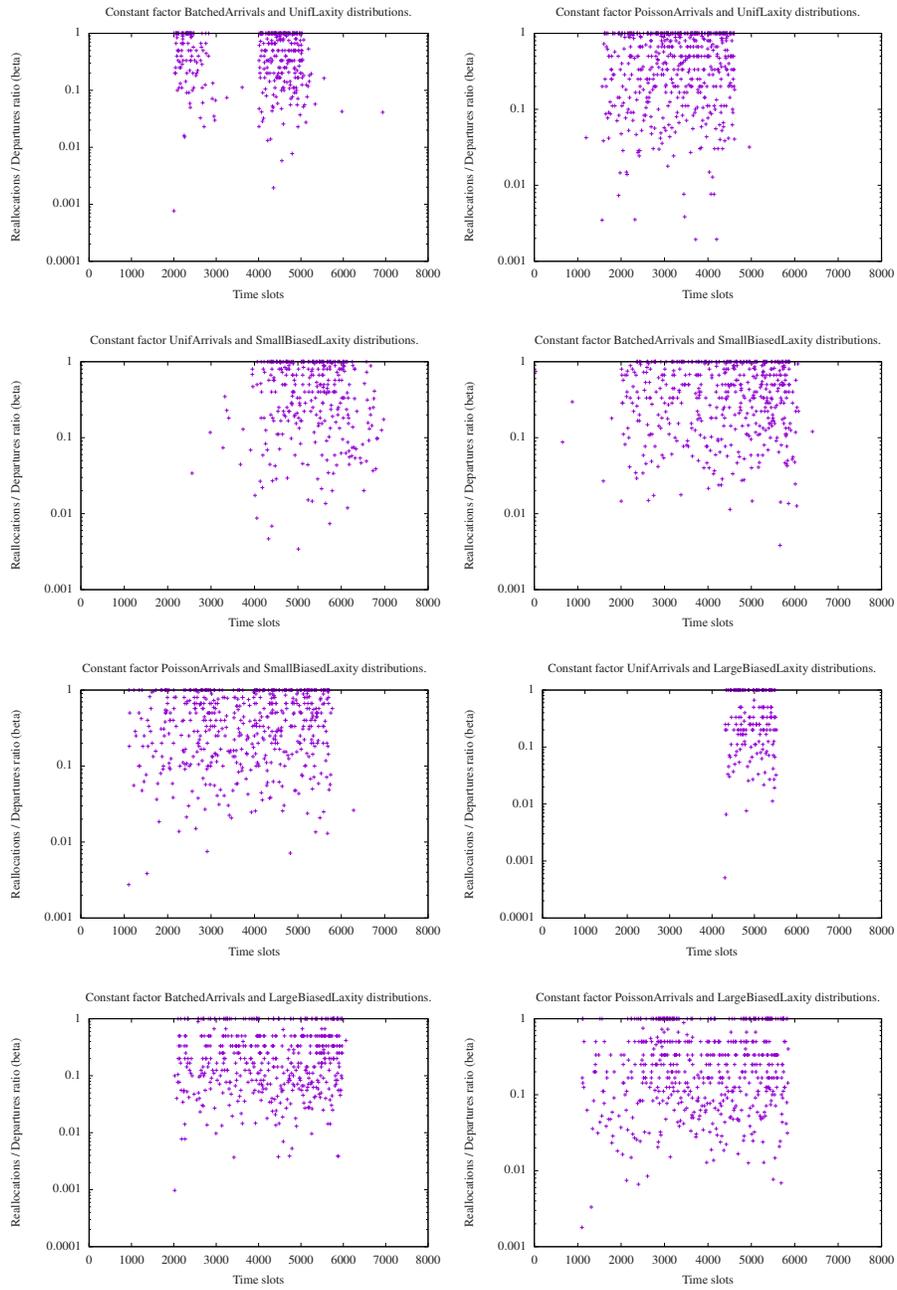

\begin{center}
%\subfigure{\resizebox{.49\textwidth}{!}{\input{epslatex_betaUnifArrivalsUnifLaxityconstant.tex}}}
\subfigure{\resizebox{.49\textwidth}{!}{\input{epslatex_betaBatchedArrivalsUnifLaxityconstant.tex}}}
\subfigure{\resizebox{.49\textwidth}{!}{\input{epslatex_betaPoissonArrivalsUnifLaxityconstant.tex}}}
\subfigure{\resizebox{.49\textwidth}{!}{\input{epslatex_betaUnifArrivalsSmallBiasedLaxityconstant.tex}}}
\subfigure{\resizebox{.49\textwidth}{!}{\input{epslatex_betaBatchedArrivalsSmallBiasedLaxityconstant.tex}}}
\subfigure{\resizebox{.49\textwidth}{!}{\input{epslatex_betaPoissonArrivalsSmallBiasedLaxityconstant.tex}}}
\subfigure{\resizebox{.49\textwidth}{!}{\input{epslatex_betaUnifArrivalsLargeBiasedLaxityconstant.tex}}}
\subfigure{\resizebox{.49\textwidth}{!}{\input{epslatex_betaBatchedArrivalsLargeBiasedLaxityconstant.tex}}}
\subfigure{\resizebox{.49\textwidth}{!}{\input{epslatex_betaPoissonArrivalsLargeBiasedLaxityconstant.tex}}}
\caption{Reallocation/Departure ratio ($\beta$) vs. time for constant classification factor, $n=4000$ and $w_{\max}=1024$.}
\label{fig:beta_constant}
\end{center}
\end{figure}

\begin{figure}[h]
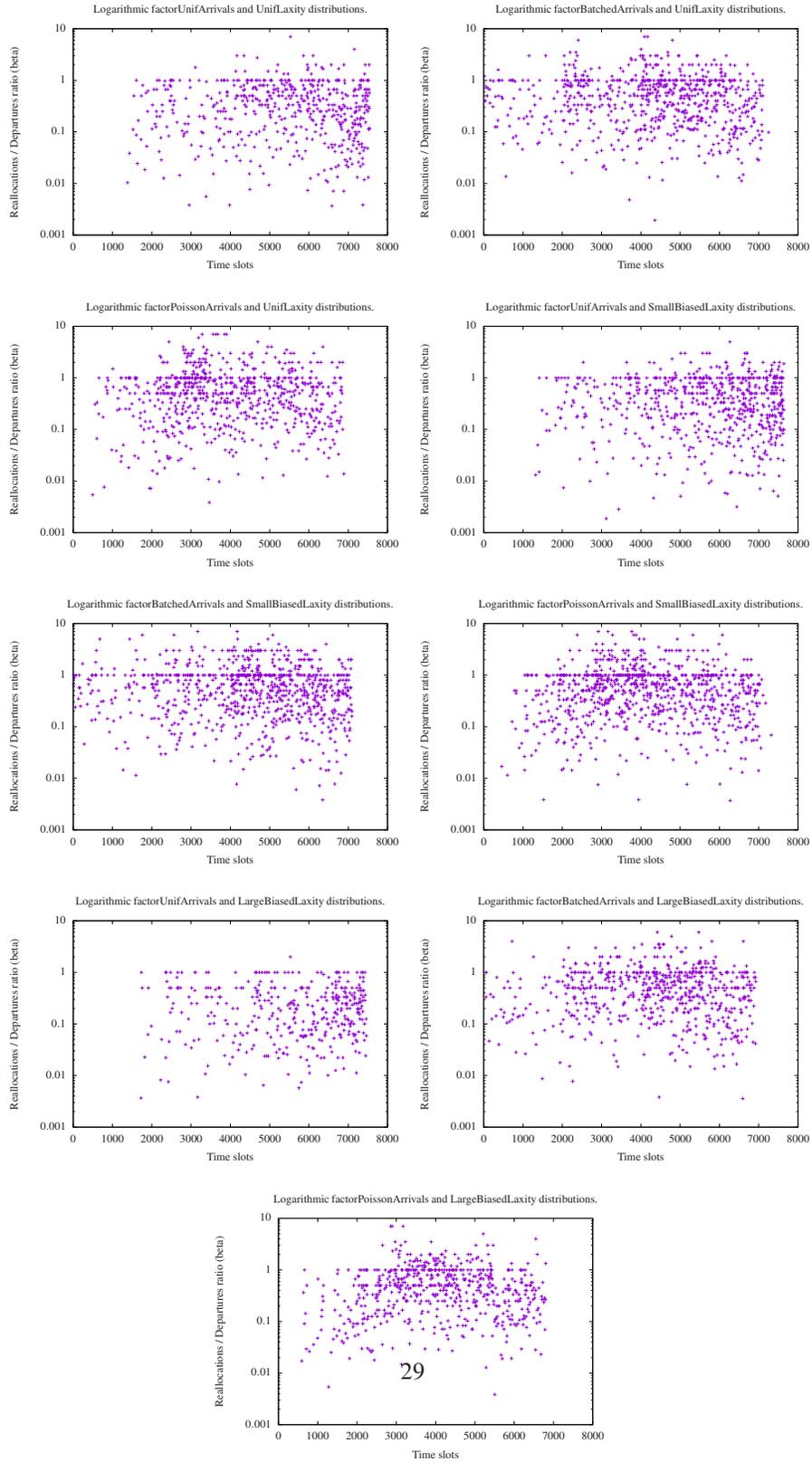

\begin{center}
\subfigure{\resizebox{.49\textwidth}{!}{\input{epslatex_betaUnifArrivalsUnifLaxitylog.tex}}}
\subfigure{\resizebox{.49\textwidth}{!}{\input{epslatex_betaBatchedArrivalsUnifLaxitylog.tex}}}
\subfigure{\resizebox{.49\textwidth}{!}{\input{epslatex_betaPoissonArrivalsUnifLaxitylog.tex}}}
\subfigure{\resizebox{.49\textwidth}{!}{\input{epslatex_betaUnifArrivalsSmallBiasedLaxitylog.tex}}}
\subfigure{\resizebox{.49\textwidth}{!}{\input{epslatex_betaBatchedArrivalsSmallBiasedLaxitylog.tex}}}
\subfigure{\resizebox{.49\textwidth}{!}{\input{epslatex_betaPoissonArrivalsSmallBiasedLaxitylog.tex}}}
\subfigure{\resizebox{.49\textwidth}{!}{\input{epslatex_betaUnifArrivalsLargeBiasedLaxitylog.tex}}}
\subfigure{\resizebox{.49\textwidth}{!}{\input{epslatex_betaBatchedArrivalsLargeBiasedLaxitylog.tex}}}
\subfigure{\resizebox{.49\textwidth}{!}{\input{epslatex_betaPoissonArrivalsLargeBiasedLaxitylog.tex}}}
\caption{Reallocation/Departure ratio ($\beta$) vs. time for logarithmic classification factor, $n=4000$ and $w_{\max}=1024$.}
\label{fig:beta_logarithmic}
\end{center}
\end{figure}

\begin{figure}[h]
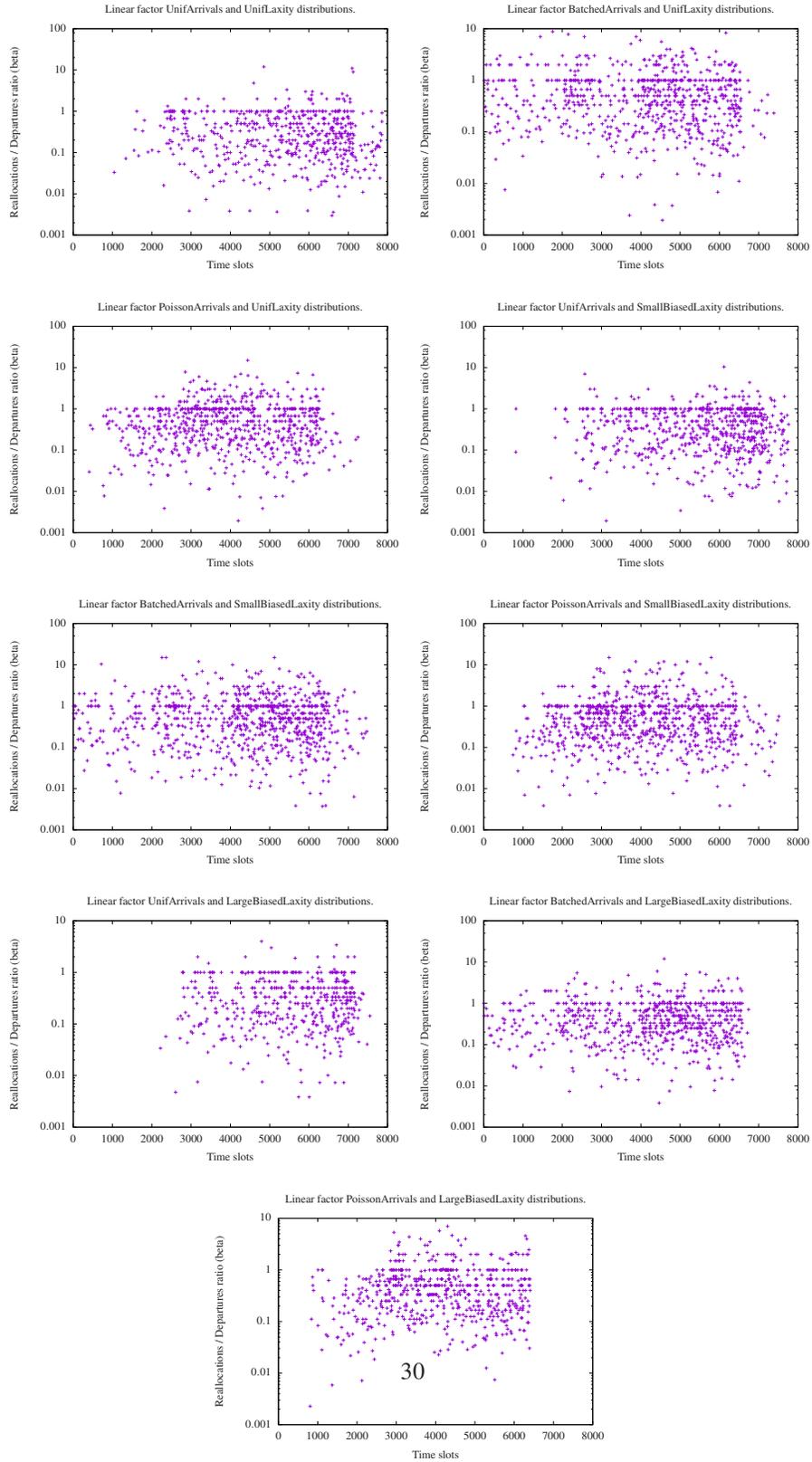

\begin{center}
\subfigure{\resizebox{.49\textwidth}{!}{\input{epslatex_betaUnifArrivalsUnifLaxitylinear.tex}}}
\subfigure{\resizebox{.49\textwidth}{!}{\input{epslatex_betaBatchedArrivalsUnifLaxitylinear.tex}}}
\subfigure{\resizebox{.49\textwidth}{!}{\input{epslatex_betaPoissonArrivalsUnifLaxitylinear.tex}}}
\subfigure{\resizebox{.49\textwidth}{!}{\input{epslatex_betaUnifArrivalsSmallBiasedLaxitylinear.tex}}}
\subfigure{\resizebox{.49\textwidth}{!}{\input{epslatex_betaBatchedArrivalsSmallBiasedLaxitylinear.tex}}}
\subfigure{\resizebox{.49\textwidth}{!}{\input{epslatex_betaPoissonArrivalsSmallBiasedLaxitylinear.tex}}}
\subfigure{\resizebox{.49\textwidth}{!}{\input{epslatex_betaUnifArrivalsLargeBiasedLaxitylinear.tex}}}
\subfigure{\resizebox{.49\textwidth}{!}{\input{epslatex_betaBatchedArrivalsLargeBiasedLaxitylinear.tex}}}
\subfigure{\resizebox{.49\textwidth}{!}{\input{epslatex_betaPoissonArrivalsLargeBiasedLaxitylinear.tex}}}
\caption{Reallocation/Departure ratio ($\beta$) vs. time for linear classification factor, $n=4000$ and $w_{\max}=1024$.}
\label{fig:beta_linear}
\end{center}
\end{figure}

The arrival time of each client was chosen: (a) uniformly at random within the interval $[1,2n]$; (b) in $3$ batches of $n/3$ clients arriving at $t=1$, $t=n/2$, and $t=n$; 
and (c) as a Poisson process with mean rate $\lambda = 0.7$. The choice of a Poisson process intends to model another case where the arrival schedule does not include bursts, whereas the value chosen for $\lambda$ intends to model an arrival schedule that is somewhat dense ($0.7$ expected arrivals per unit of time until all $n$ clients have arrived).
For each client, the departure time was chosen uniformly at random from the interval $[t_a,2n]$, where $t_a$ is the time of arrival of such client.
The inputs for $n=4000$ and $w_{\max}=1024$ are illustrated in Figure~\ref{fig:H} %in the Appendix 
showing the $H(C(t))$ function, which is a lower bound on the optimal number of stations needed.

With respect to the protocol, three different classification factors: constant, logarithmic, and linear, were used, as detailed in Algorithm~\ref{alg:classifier}. %in the Appendix. 
We implemented the protocol and input generator in Java 8. The simulations were carried out on one of the Linux servers at Pace University. The specifications are 
Intel\circledR Xeon\circledR CPU X5450  @ 3.00GHz, 2GB RAM, 150GB HD, running Debian 8 x64.

For each of the $243$ scenarios that arise from the combination of the above variants (3 $w_{\max}$, 3 laxity distributions, 3 arrival distributions, 3 numbers of clients, and 3 protocols), 
we evaluated experimentally the $(\alpha,\beta)$-\name of CPR. 
Our simulations showed that the performance in practical settings is indeed as expected or better than the theoretical bounds (as in Corollary~\ref{cor:gral}). 
The discussion and plots that follow, refer to $n=4000$ and $w_{\max}=1024$, but similar results were obtained for the other cases. 
The source code, the input data, and the raw output data are publicly available in~\cite{simulator}.

%Indeed, it can be seen in Figures~\ref{fig:beta_constant},~\ref{fig:beta_logarithmic}, and~\ref{fig:beta_linear} that the reallocation vs. departures weight ratio (bounded by $\beta$) is zero for many time slots, which is expected because in many time slots no client arrives or departs,} and frequently at most $1$ for the time slots where some clients are reallocated. 
It can be seen in Figures~\ref{fig:beta_constant},~\ref{fig:beta_logarithmic}, and~\ref{fig:beta_linear} that the reallocation vs. departures weight ratio (bounded by $\beta$) is frequently at most $1$.
For constant factor classification on uniform arrival distribution and uniform laxity no client was ever reallocated. Hence, this case is not plotted. Also, the ratio is defined on reallocation events. Hence, no data points are shown in time slots without reallocations.

% \mig{[MM: check now. Updating the corresponding response to reviewer is pending]}

%To quantify the latter observations we focus on the time slots where some client has departed, and we compute statistics of the reallocation vs. departures weight ratio over those slots. 
%The results are shown in Figure~\ref{fig:betabars} where we see in the top plot that on average the ratio is always below $0.3$. For comparison, notice that the bounds $\beta$ proved in Corollary~\ref{cor:part} for $\rho=1$ and $w_{\max}=2$ are $3$, $1.828$, and $1$, for constant, logarithmic, and linear classification factors respectively. 
%That is, as seen in the plot, much higher than the average and higher than the average plus one standard deviation, and the difference is larger for $w_{\max}>2$.
%Moreover, in the bottom plot of the same figure we see that similar observations hold even if we consider only time slots where some client is reallocated.
To quantify the latter observations, we compute statistics of the reallocation vs. departures weight ratio, over time slots where some client has been reallocated. 
The results are shown in Figure~\ref{fig:betabars}. It can be seen that in all cases the average plus one standard deviation is below $2.5$. 
For comparison, we compute the bounds $\beta$ proved in Corollary~\ref{cor:gral}. 
%The corollary holds for $w_{\max}>w_{\min}\geq 4$ and for laxities that are powers of $2$. 
Recall that the sample space for $w_{\max}$ in the simulations was $[1,1024]$. 
Nevertheless, being pessimistic and replacing $w_{\max}=8$, and the value $\rho=1$ used in our simulations, we have that the theoretical upper bound is $\beta\geq 3$ for all classification factors. 
For larger values of $w_{\max}$ the gap between our observations and the theoretical bound is even larger, showing that on realistic inputs our protocol behaves much better than \pru{the} worst-case theoretical bounds.

With respect to station usage, Figure~\ref{fig:alpha} %in the Appendix 
shows that after a period upon initial arrivals and a period before last departures, the station usage ratio against $H(C(t))$, which is only a lower bound of the optimal, (bounded by $\alpha$) is most of the time below $4$, and frequently below $2$.
We make this observation more precise \pru{by} computing the percentage of time slots when the station usage ratio against $H(C(t))$ is below $4$ for each combination of classification factor and arrival distribution. The results are shown in Table~\ref{table:percent}.

\begin{figure*}[h]
\begin{center}
\subfigure{\resizebox{.49\textwidth}{!}{\input{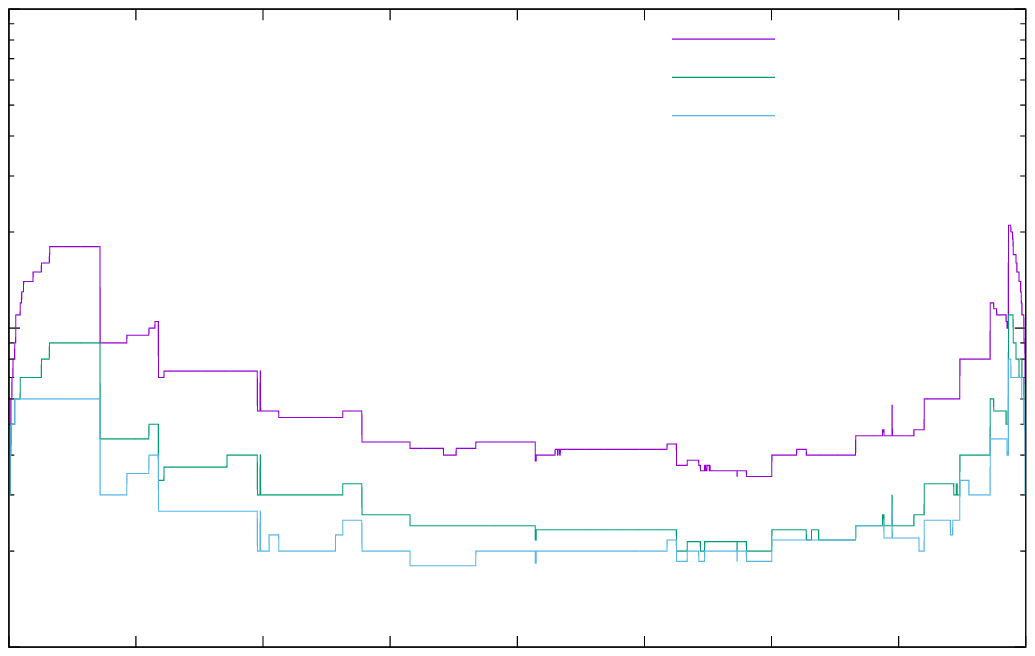}}}
\subfigure{\resizebox{.49\textwidth}{!}{\input{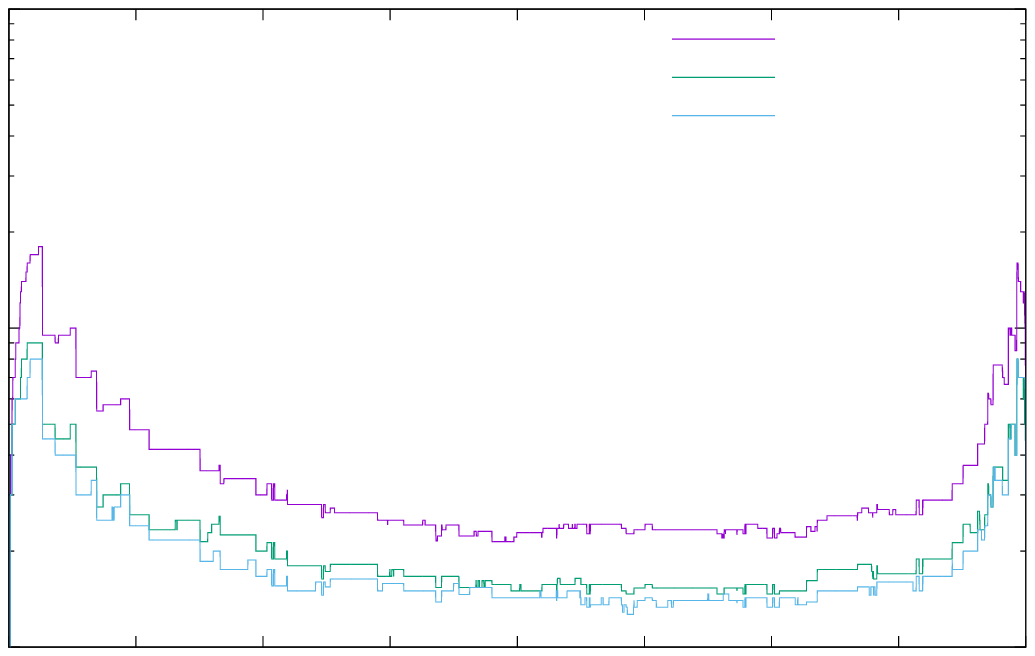}}}
\subfigure{\resizebox{.49\textwidth}{!}{\input{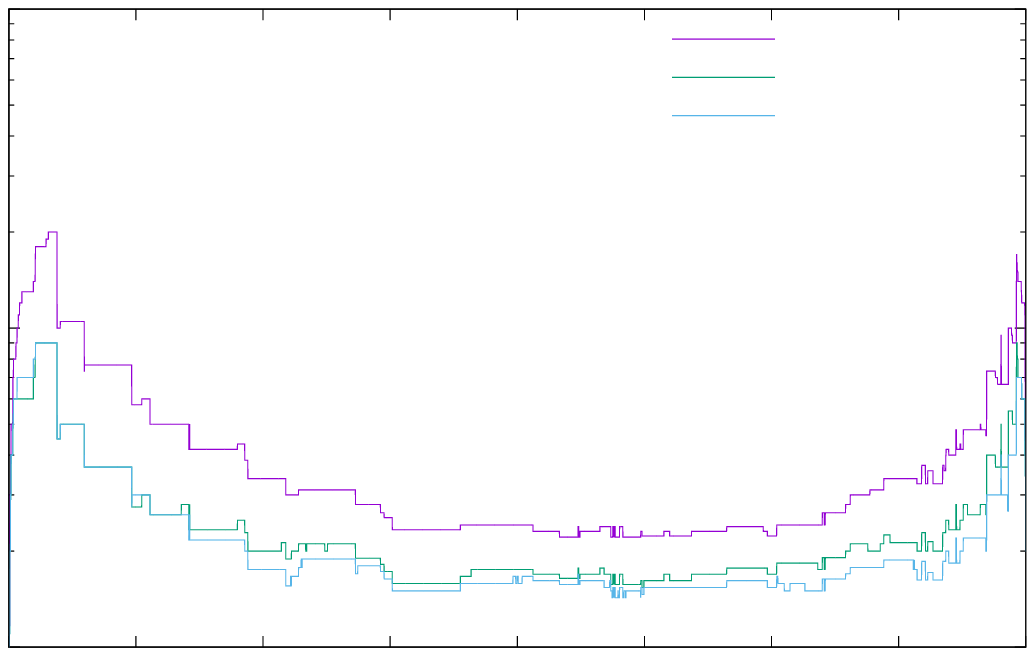}}}
\subfigure{\resizebox{.49\textwidth}{!}{\input{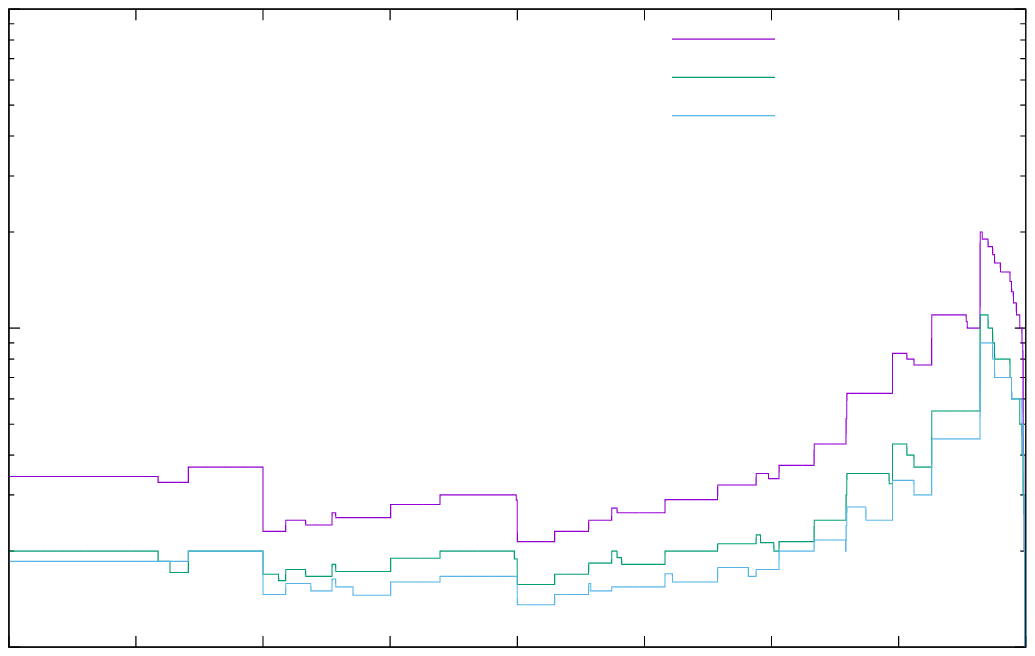}}}
\subfigure{\resizebox{.49\textwidth}{!}{\input{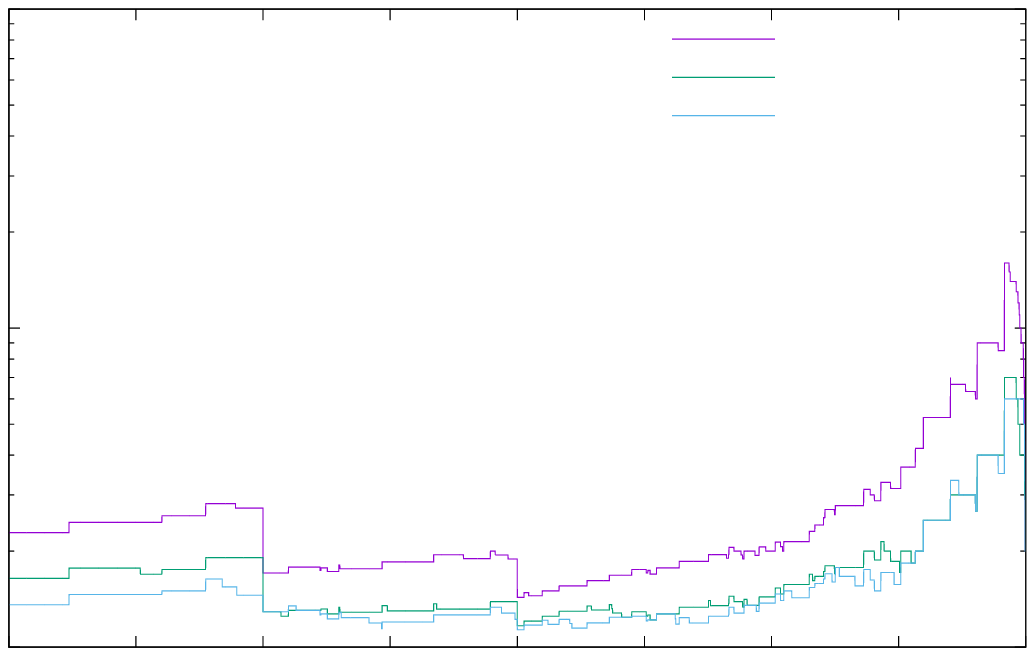}}}
\subfigure{\resizebox{.49\textwidth}{!}{\input{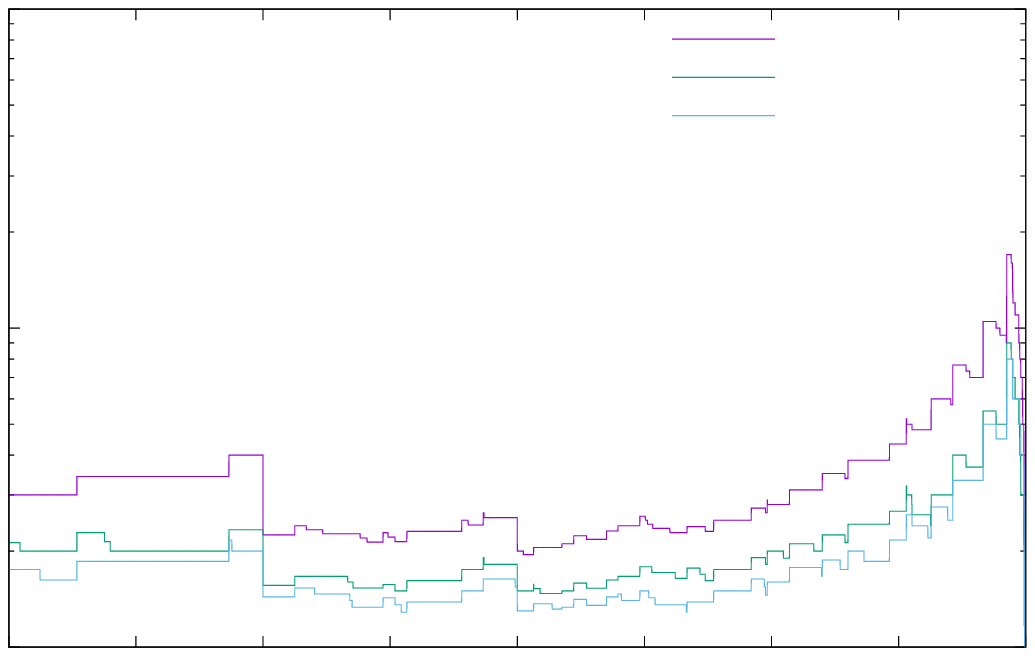}}}
\subfigure{\resizebox{.49\textwidth}{!}{\input{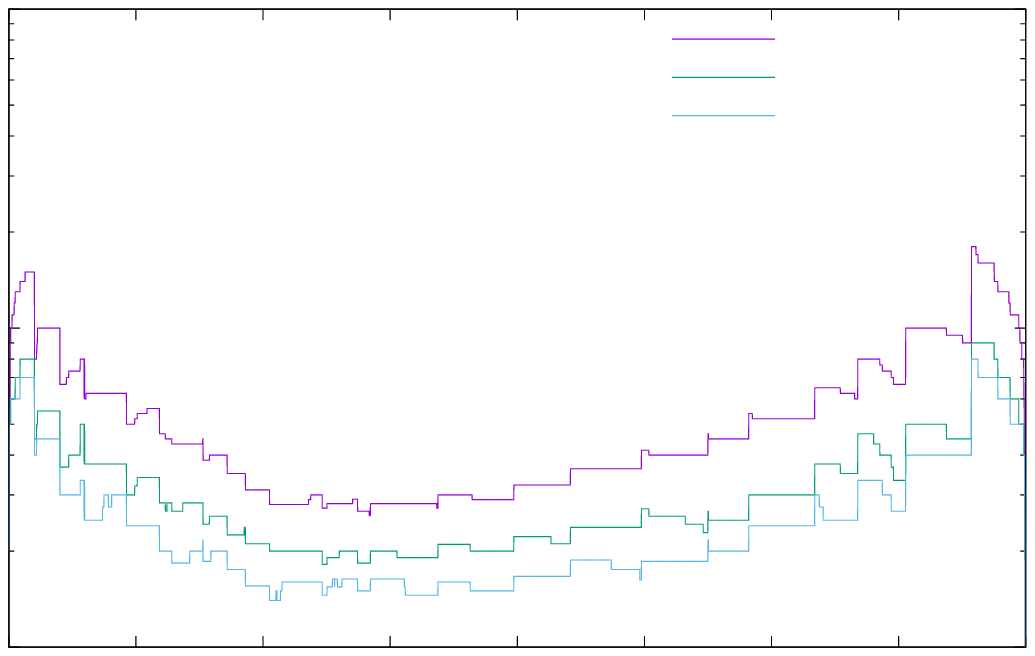}}}
\subfigure{\resizebox{.49\textwidth}{!}{\input{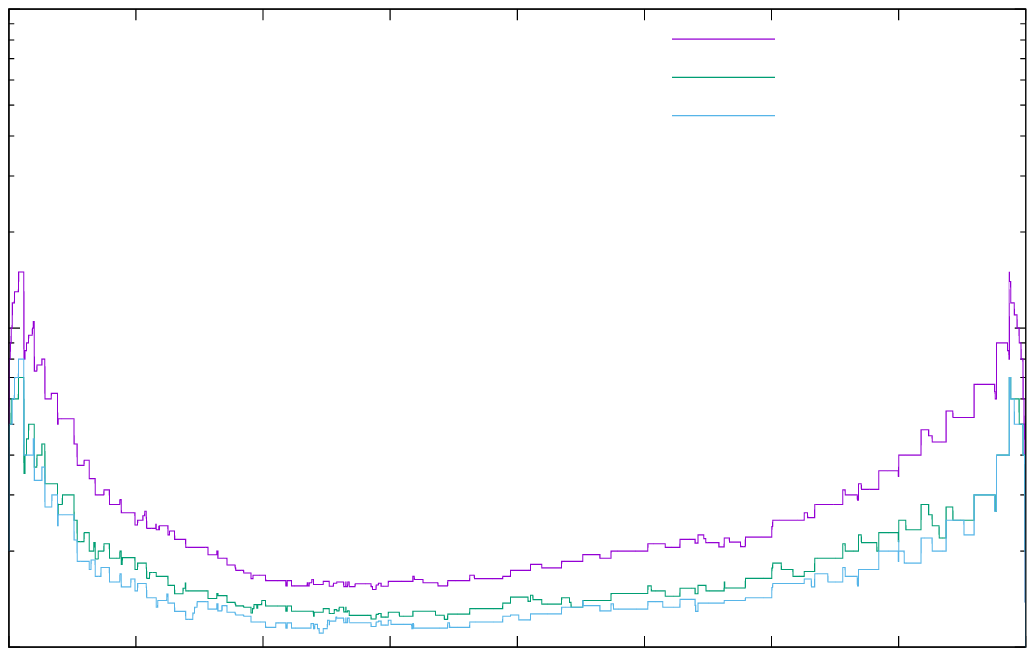}}}
\subfigure{\resizebox{.49\textwidth}{!}{\input{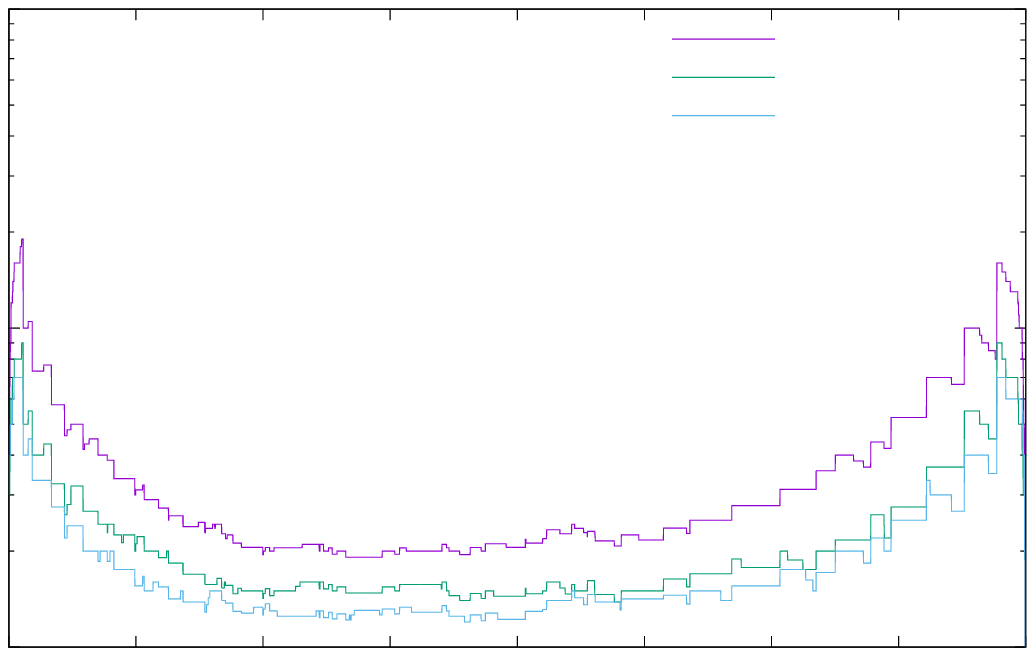}}}
\caption{Station usage ratio ($\alpha$) vs. time,  for $n=4000$ and $w_{\max}=1024$.}
\label{fig:alpha}
\end{center}
\end{figure*}

%\begin{table}%
%\centering
%\begin{tabular}{|c|c|c|c|}
%\hline
%Laxity distribution & Arrival distribution & Factor & percentage \\
%\hline
%Uniform &Uniform &Constant & $0.0375$\\
%Small Biased &Uniform &Constant & $0.025$\\
%Large Biased &Uniform &Constant & $0.0$\\
%Uniform &Batched &Constant & $1.0375$\\
%Small Biased &Batched &Constant & $46.0375$\\
%Large Biased &Batched &Constant & $0.0125$\\
%Uniform &Poisson &Constant & $5.675$\\
%Small Biased &Poisson &Constant & $39.5625$\\
%Large Biased &Poisson &Constant & $0.0125$\\
%Uniform &Uniform &Logarithmic & $48.7875$\\
%Small Biased &Uniform &Logarithmic & $66.7375$\\
%Large Biased &Uniform &Logarithmic & $0.0$\\
%Uniform &Batched &Logarithmic & $50.325$\\
%Small Biased &Batched &Logarithmic & $86.1$\\
%Large Biased &Batched &Logarithmic & $34.6875$\\
%Uniform &Poisson &Logarithmic & $63.775$\\
%Small Biased &Poisson &Logarithmic & $72.3125$\\
%Large Biased &Poisson &Logarithmic & $6.975$\\
%Uniform &Uniform &Linear & $69.7125$\\
%Small Biased &Uniform &Linear & $74.375$\\
%Large Biased &Uniform &Linear & $10.7125$\\
%Uniform &Batched &Linear & $81.8125$\\
%Small Biased &Batched &Linear & $89.1375$\\
%Large Biased &Batched &Linear & $68.4125$\\
%Uniform &Poisson &Linear & $72.1625$\\
%Small Biased &Poisson &Linear & $80.5625$\\
%Large Biased &Poisson &Linear & $49.825$\\
%\hline
%\end{tabular}
%\caption{Percentage of time slots when the station usage ratio is below $2$, for each classification factor, laxity distribution, and arrival distribution}
%\label{table:percent}
%\end{table}%

\begin{table}%
\centering
\begin{tabular}{|c|c|c|c|}
\hline
Laxity distribution & Arrival distribution & Factor & percentage \\
\hline
Unif & Unif & Const & $69.0875$\\
SmallBiased & Unif & Const & $76.5625$\\
LargeBiased & Unif & Const & $9.575$\\
Unif & Batched & Const & $83.4$\\
SmallBiased & Batched & Const & $89.225$\\
LargeBiased & Batched & Const & $79.3125$\\
Unif & Poisson & Const & $73.3875$\\
SmallBiased & Poisson & Const & $80.9$\\
LargeBiased & Poisson & Const & $41.475$\\
Unif & Unif & Log & $90.1875$\\
SmallBiased & Unif & Log & $91.9375$\\
LargeBiased & Unif & Log & $75.925$\\
Unif & Batched & Log & $95.0$\\
SmallBiased & Batched & Log & $95.3375$\\
LargeBiased & Batched & Log & $88.825$\\
Unif & Poisson & Log & $90.0625$\\
SmallBiased & Poisson & Log & $94.05$\\
LargeBiased & Poisson & Log & $78.4375$\\
Unif & Unif & Linear & $91.05$\\
SmallBiased & Unif & Linear & $91.9375$\\
LargeBiased & Unif & Linear & $86.725$\\
Unif & Batched & Linear & $96.0$\\
SmallBiased & Batched & Linear & $95.925$\\
LargeBiased & Batched & Linear & $90.875$\\
Unif & Poisson & Linear & $92.5875$\\
SmallBiased & Poisson & Linear & $94.7375$\\
LargeBiased & Poisson & Linear & $83.275$\\
\hline
\end{tabular}
\caption{Percentage of time slots when the station usage ratio is below $4$, for each classification factor, laxity distribution, and arrival distribution,  for $n=4000$ and $w_{\max}=1024$.}
\label{table:percent}
\end{table}%

Should the reallocation ratio be minimized, the constant factor classification achieves better performance at a higher station usage. On the other hand, if station usage must be kept low, the linear factor classification performs better incurring in higher reallocation cost. The logarithmic factor balances both costs. Figure~\ref{fig:alphavsbeta} illustrates these trade offs. In comparison with the bounds proved in Corollary~\ref{cor:gral}, for the scenarios simulated CPR behaves better than expected.
As we see in the figure, these trade-offs appear in all input distributions, although in some the impact is milder (e.g. large-biased laxities with uniform or Poisson arrivals).

The inputs chosen for our evaluation are intuitively representative of a variety of likely cases. Namely, bursts and smooth arrivals, more/even/less demanding clients, etc. Should a comparison among factors regardless of distributions be needed (e.g., if the distribution is unknown, but the extremal values of bandwidths, laxities, and $H(C(t))$ are known) the worst-case guarantees in the analysis must be used.

\section{Conclusions and Future Work}
\label{sec:conclude}

\pru{
In this paper, we study a dynamic allocation problem SA and associated reallocation algorithms
assuming that clients have laxity and bandwidth requirements.
We characterize these algorithms by defining the notion $(\alpha,\beta)$-performance as combination
of the competitive ratio on station usage ($\alpha$) and the cost of reallocations ($\beta$).
We show that previous protocols that work well for unit cost per client reallocation
do not work well when the cost is more general.
We then present a new protocol called Classified Preemptive Reallocation and 
prove bounds on both of our performance metrics.
We also present experimental simulation results on average cases supplementing our theoretical analysis on worst case.
}

\pru{
There are a few future directions.
To further understand the performance of algorithms,
it is desirable to derive lower bounds on the performance ratio of a general algorithms.
In this paper we assume that each station has the same capacity.
An obvious generalization is to consider stations having different capacities.
In addition, we may extend cost model to introduce a weight to each client and the reallocation cost is then calculated as a weighted cost.
In terms of the setting, we aim to quantify the resources required to complete all requests from clients.
A direction is to consider limited resources and striking a balance between completing more clients and not violating the resource limitation.
}

\section*{Acknowledgements}
\pru{
The authors thank the support from a Visiting Fellowship and the initiative Networks Sciences \& Technologies (NeST)
by School of EEE \& CS, University of Liverpool,} \mig{as well as Pace University NYFC SRC Award and Kenan Fund Award.}

%\clearpage

\bibliographystyle{plain}
\bibliography{references2}

\begin{thebibliography}{10}

\bibitem{AdamyE03}
Udo Adamy and Thomas Erlebach.
\newblock Online coloring of intervals with bandwidth.
\newblock In {\em Proceedings of the 1st International Workshop on
  Approximation and Online Algorithms}, volume 2909 of {\em Lecture Notes in
  Computer Science}, pages 1--12. Springer, 2003.

\bibitem{AlF07}
Susanne Albers and Hiroshi Fujiwara.
\newblock Energy-efficient algorithms for flow time minimization.
\newblock {\em ACM Transactions on Algorithms}, 3(4):49, 2007.

\bibitem{AlbersH12}
Susanne Albers and Matthias Hellwig.
\newblock On the value of job migration in online makespan minimization.
\newblock In {\em Proceedings of the 20th Annual European Symposium on
  Algorithms}, volume 7501 of {\em Lecture Notes in Computer Science}, pages
  84--95. Springer, 2012.

\bibitem{DBLP:conf/focs/AndrewsAZ10}
Matthew Andrews, Spyridon Antonakopoulos, and Lisa Zhang.
\newblock Minimum-cost network design with (dis)economies of scale.
\newblock In {\em Proceedings of the 51st Annual {IEEE} Symposium on
  Foundations of Computer Science}, pages 585--592. {IEEE} Computer Society,
  2010.

\bibitem{Aza97}
Yossi Azar.
\newblock On-line load balancing.
\newblock In {\em Proceedings of Developments from a June 1996 Seminar on
  Online Algorithms: The State of the Art}, pages 178--195. Springer-Verlag,
  1996.

\bibitem{azarFracmatching}
Yossi Azar and Arik Litichevskey.
\newblock Maximizing throughput in multi-queue switches.
\newblock {\em Algorithmica}, 45:69--90, 2006.

\bibitem{BaloghBG10}
J{\'{a}}nos Balogh, J{\'{o}}zsef B{\'{e}}k{\'{e}}si, and G{\'{a}}bor Galambos.
\newblock New lower bounds for certain classes of bin packing algorithms.
\newblock In {\em Proceedings of the 8th International Workshop on
  Approximation and Online Algorithms (WAOA)}, pages 25--36, 2010.

\bibitem{BansalCP13}
Nikhil Bansal, Ho-Leung Chan, and Kirk Pruhs.
\newblock Speed scaling with an arbitrary power function.
\newblock {\em ACM Transactions on Algorithms}, 9(2):18:1--18:14, 2013.

\bibitem{NPh}
Amotz Bar-Noy, Randeep Bhatia, Joseph Naor, and Baruch Schieber.
\newblock Minimizing service and operation costs of periodic scheduling.
\newblock In {\em Proceedings of the 9th Annual ACM-SIAM Symposium on Discrete
  Algorithms}, pages 11--20, 1998.

\bibitem{bar2003windows}
Amotz Bar-Noy and Richard~E Ladner.
\newblock Windows scheduling problems for broadcast systems.
\newblock {\em SIAM Journal on Computing}, 32(4):1091--1113, 2003.

\bibitem{bar2007windows}
Amotz Bar-Noy, Richard~E Ladner, and Tami Tamir.
\newblock Windows scheduling as a restricted version of bin packing.
\newblock {\em ACM Transactions on Algorithms}, 3(3):28, 2007.

\bibitem{BG04}
Sanjoy Baruah and Jo\"el Goossens.
\newblock Scheduling real-time tasks: Algorithms and complexity.
\newblock In J.~Leung, editor, {\em Handbook of Scheduling: Algorithms, Models
  and Performance Analysis}, pages 15--1--15--41. CRC Press, 2004.

\bibitem{leonardi}
L.~Becchetti, S.~Leonardi, A.~Marchetti-Spaccamela, A.~Vitaletti, S.~Diggavi,
  S.~Muthukrishnan, and T.~Nandagopal.
\newblock Parallel scheduling problems in next generation wireless networks.
\newblock {\em Networks}, 45(1):9--22, 2005.

\bibitem{bender2013reallocation}
Michael~A. Bender, Mart\'in Farach-Colton, S{\'a}ndor~P. Fekete, Jeremy~T.
  Fineman, and Seth Gilbert.
\newblock Reallocation problems in scheduling.
\newblock In {\em Proceedings of the 25th {ACM} Symposium on Parallelism in
  Algorithms and Architectures}, pages 271--279. {ACM}, 2013.

\bibitem{ChanCLLMW09}
Ho{-}Leung Chan, Joseph~Wun{-}Tat Chan, Tak~Wah Lam, Lap{-}Kei Lee, Kin{-}Sum
  Mak, and Prudence W.~H. Wong.
\newblock Optimizing throughput and energy in online deadline scheduling.
\newblock {\em {ACM} Transactions on Algorithms}, 6(1):1--22, 2009.

\bibitem{chan2005temporary}
Wun-Tat Chan and PrudenceW.H. Wong.
\newblock On-line windows scheduling of temporary items.
\newblock In {\em Proceedings of the 15th International Symposium on Algorithms
  and Computation}, volume 3341 of {\em Lecture Notes in Computer Science},
  pages 259--270. Springer, 2004.

\bibitem{CoffmanGS1996}
E.~G. {Coffman, Jr}., M.~R. Garey, and D.~S. Johnson.
\newblock Bin packing approximation algorithms: A survey.
\newblock In Dorit~S. Hochbaum, editor, {\em Approximation Algorithms for
  NP-Hard Problems}, pages 46--93. PWS, 1996.

\bibitem{coffman2013bin}
Edward~G Coffman~Jr, J{\'a}nos Csirik, G{\'a}bor Galambos, Silvano Martello,
  and Daniele Vigo.
\newblock Bin packing approximation algorithms: survey and classification.
\newblock In {\em Handbook of Combinatorial Optimization}, pages 455--531.
  Springer, 2013.

\bibitem{CoffmanGMV1998}
Edward~G. {Coffman, Jr}., Gabor Galambos, Silvano Martello, and Daniele Vigo.
\newblock Bin packing approximation algorithms: Combinatorial analysis.
\newblock In D.-Z. Du and P.~M. Pardalos, editors, {\em Handbook of
  Combinatorial Optimization}, pages 151--207. Kluwer Academic Publishers,
  1998.

\bibitem{CominardiGBO17}
Luca Cominardi, Fabio Giust, Carlos~Jesus Bernardos, and Antonio de~la Oliva.
\newblock Distributed mobility management solutions for next mobile network
  architectures.
\newblock {\em Computer Networks}, 121:124--136, 2017.

\bibitem{Epstein10}
Leah Epstein.
\newblock Bin packing with rejection revisited.
\newblock {\em Algorithmica}, 56(4):505--528, 2010.

\bibitem{EpsteinEL09}
Leah Epstein, Thomas Erlebach, and Asaf Levin.
\newblock Variable sized online interval coloring with bandwidth.
\newblock {\em Algorithmica}, 53(3):385--401, 2009.

\bibitem{Farach-ColtonLMT14}
Mart\'in Farach-Colton, Katia Leal, Miguel~A. Mosteiro, and Christopher
  Thraves.
\newblock Dynamic windows scheduling with reallocation.
\newblock In {\em Proceedings of the 13th International Symposium on
  Experimental Algorithms}, volume 8504 of {\em Lecture Notes in Computer
  Science}, pages 99--110. Springer, 2014.

\bibitem{muthuAdwords}
Jon Feldman, Aranyak Mehta, Vahab Mirrokni, and S.~Muthukrishnan.
\newblock Online stochastic matching: Beating 1-1/e.
\newblock In {\em Proceedings of the 50th Annual {IEEE} Symposium on
  Foundations of Computer Science}, pages 117 --126. {IEEE} Computer Society,
  2009.

\bibitem{AntaKMW13}
Antonio {Fern{\'{a}}ndez Anta}, Dariusz~R. Kowalski, Miguel~A. Mosteiro, and
  Prudence W.~H. Wong.
\newblock Station assignment with applications to sensing.
\newblock In {\em Proceedings of the 9th International Symposium on Algorithms
  and Experiments for Sensor Systems, Wireless Networks and Distributed
  Robotics}, volume 8243 of {\em Lecture Notes in Computer Science}, pages
  155--169. Springer, 2013.

\bibitem{PruhsPrimalDual2012}
Anupam Gupta, Ravishankar Krishnaswamy, and Kirk Pruhs.
\newblock Online primal-dual for non-linear optimization with applications to
  speed scaling.
\newblock In {\em Proceedings of the 10th Workshop on Approximation and Online
  Algorithms}, volume 7846 of {\em Lecture Notes in Computer Science}, pages
  173--186. Springer, 2012.

\bibitem{simulator}
Austin Halper, Miguel~A. Mosteiro, Yulia Rossikova, and Prudence W.~H. Wong.
\newblock Station assignment with reallocation simulator code and data.
\newblock \burl{http://csis.pace.edu/~mmosteiro/pub/sourceBSreallocJournal/},
  2017.

\bibitem{HolteMR+89}
Robert Holte, Al~Mok, Louis Rosier, Igor Tulchinsky, and Donald Varvel.
\newblock The pinwheel: a real-time scheduling problem.
\newblock In {\em Proceedings of the 22nd Annual Hawaii International
  Conference on System Sciences}, volume II, Software Track, pages 693--702,
  1989.

\bibitem{JacobsL14}
Tobias Jacobs and Salvatore Longo.
\newblock A new perspective on the windows scheduling problem.
\newblock {\em CoRR}, abs/1410.7237, 2014.

\bibitem{JietalToN2015}
Bo~Ji, Gagan~R. Gupta, Manu Sharma, Xiaojun Lin, and Ness~B. Shroff.
\newblock Achieving optimal throughput and near-optimal asymptotic delay
  performance in multi-channel wireless networks with low complexity: A
  practical greedy scheduling policy.
\newblock {\em IEEE/ACM Transactions on Networking}, 23(3):880--893, 2015.

\bibitem{pruhsBmatching}
Bala Kalyanasundaram and Kirk Pruhs.
\newblock An optimal deterministic algorithm for online b-matching.
\newblock {\em Theoretical Computer Science}, 233(1-2):319--325, 2000.

\bibitem{KanjoBPCFWCW08}
Eiman Kanjo, Steve Benford, Mark Paxton, Alan Chamberlain, Danae~Stanton
  Fraser, Dawn Woodgate, David Crellin, and Adrian Woolard.
\newblock Mobgeosen: facilitating personal geosensor data collection and
  visualization using mobile phones.
\newblock {\em Personal and Ubiquitous Computing}, 12(8):599--607, 2008.

\bibitem{KhanXAA13}
Wazir~Zada Khan, Yang Xiang, Mohammed~Y. Aalsalem, and Quratul{-}Ain Arshad.
\newblock Mobile phone sensing systems: {A} survey.
\newblock {\em {IEEE} Communications Surveys and Tutorials}, 15(1):402--427,
  2013.

\bibitem{MosteiroRW15}
Miguel~A. Mosteiro, Yulia Rossikova, and Prudence~W.H. Wong.
\newblock Station assignment with reallocation.
\newblock In {\em Proceedings of the 14th International Symposium on
  Experimental Algorithms}, Lecture Notes in Computer Science, pages 151--164.
  Springer, 2015.

\bibitem{RestucciaDP16}
Francesco Restuccia, Sajal~K. Das, and Jamie Payton.
\newblock Incentive mechanisms for participatory sensing: Survey and research
  challenges.
\newblock {\em {TOSN}}, 12(2):13:1--13:40, 2016.

\bibitem{SandersSS04}
Peter Sanders, Naveen Sivadasan, and Martin Skutella.
\newblock Online scheduling with bounded migration.
\newblock In {\em Proceedings of the 31st International Colloquium on Automata,
  Languages and Programming}, volume 3142 of {\em Lecture Notes in Computer
  Science}, pages 1111--1122. Springer, 2004.

\bibitem{SandersSS09}
Peter Sanders, Naveen Sivadasan, and Martin Skutella.
\newblock Online scheduling with bounded migration.
\newblock {\em Math. Oper. Res.}, 34(2):481--498, 2009.

\bibitem{ShaZS+08}
Kewei Sha, Guoxing Zhan, Weisong Shi, Mark Lumley, Clairy Wiholm, and Bengt
  Arnetz.
\newblock Spa: A smart phone assisted chronic illness self-management system
  with participatory sensing.
\newblock In {\em Proceedings of the 2Nd International Workshop on Systems and
  Networking Support for Health Care and Assisted Living Environments},
  HealthNet '08, pages 5:1--5:3, New York, NY, USA, 2008. ACM.

\end{thebibliography}

%% Non-BibTeX users please use
%\begin{thebibliography}{}
%%
%% and use \bibitem to create references. Consult the Instructions
%% for authors for reference list style.
%\end{thebibliography}

\end{document}